\documentclass[aps, pra, reprint, superscriptaddress, onecolumn, notitlepage, tightenlines, 11pt]{revtex4-2}

\pdfoutput=1

\usepackage{header}

\graphicspath{{graphics/}}

\begin{document}

\title{Fast Algorithms for Stoquastic Spin Systems}

\author{Ryan L. Mann}
\email{mail@ryanmann.org}
\homepage{http://www.ryanmann.org}
\affiliation{Centre for Quantum Software and Information, School of Computer Science, Faculty of Engineering \& Information Technology, University of Technology Sydney, NSW 2007, Australia}

\begin{abstract}
    We establish a general framework for developing fast sampling and counting algorithms for stoquastic spin systems at high temperature. Our framework is based on a rapidly mixing Markov chain for polymer models and a subcritical percolation process for sampling individual polymers. We apply our framework to obtain fast algorithms for approximating the partition function and sampling from the thermal distribution of (1) general stoquastic spin systems, (2) ferromagnetic Heisenberg models, and (3) antiferromagnetic Heisenberg models on bipartite graphs. For the Heisenberg models, we obtain an improved bound on the inverse temperature by using their respective cycle and loop representations.
\end{abstract}

\maketitle

{
\hypersetup{linkcolor=black}
\tableofcontents
}

\section{Introduction}
\label{section:Introduction}

Approximating the partition function and sampling from the thermal distribution are central problems in statistical physics. The polymer model formalism provides a general framework for developing efficient algorithms for these problems~\cite{gruber1971general, kotecky1986cluster}. The principal algorithmic technique within this framework is the cluster expansion, which has been successfully applied to obtain efficient algorithms for the hardcore model~\cite{helmuth2020algorithmic, jenssen2020algorithms}, the Potts model~\cite{helmuth2020algorithmic, jenssen2020algorithms, borgs2020efficient}, and quantum spin systems~\cite{mann2021efficient, helmuth2023efficient}. The resulting algorithms run in polynomial time, though the degree of the polynomial is often relatively high. For classical systems, fast algorithms have been developed via Markov chain methods for polymer models~\cite{chen2021fast, galanis2022fast, blanca2024fast}.

These fast classical algorithms are based on sampling directly from the polymer model rather than explicitly computing a truncated cluster expansion. In particular, the polymer dynamics of Chen et al.~\cite{chen2021fast} provides a rapidly mixing Markov chain for sampling polymer configurations, while the graphlet sampling algorithm of Blanca et al.~\cite{blanca2024fast} provides an efficient method for sampling individual polymers via a subcritical percolation process. These methods apply broadly to classical systems, raising the question of whether they can be extended to quantum systems.

Classical algorithms for quantum systems are known in several settings. The cluster expansion has been used to obtain efficient algorithms for general quantum systems at high temperature~\cite{mann2021efficient, yin2023polynomial, mann2024algorithmic, bakshi2024high, ramkumar2026high}, stable quantum perturbations of classical spin systems at low temperature~\cite{helmuth2023efficient}, and weakly-interacting quantum spin systems at arbitrary temperature~\cite{mann2026efficient}. More efficient algorithms have been obtained via importance sampling of the cumulant expansion for weakly-interacting fermionic systems at arbitrary temperature~\cite{chen2025convergence}, and via importance sampling of the cluster expansion for general quantum systems at high temperature~\cite{putterman2026quantum}.

Quantum algorithms for thermal state preparation provide an alternative approach to sampling from the thermal distribution~\cite{chen2023efficient, ding2025efficient, gilyen2024quantum}. This approach has been used to obtain efficient algorithms for general quantum spin systems at high temperature~\cite{rouze2026optimal, bergamaschi2026fast}, weakly-interacting quantum spin systems at arbitrary temperature~\cite{smid2025rapid}, and weakly-interacting fermionic systems at arbitrary temperature~\cite{tong2025fast, smid2025polynomial}. These quantum algorithms are typically faster than their classical counterparts, motivating the search for faster classical algorithms for quantum systems.

An important class of quantum spin systems is that of stoquastic spin systems~\cite{bravyi2008complexity}, characterised by Hamiltonians whose interaction terms have non-positive off-diagonal matrix elements in the spin basis. This class includes the transverse-field Ising model, the ferromagnetic Heisenberg model, and the antiferromagnetic Heisenberg model on bipartite graphs. These systems are amenable to Markov chain methods~\cite{suzuki1977monte}, and efficient classical algorithms are known for transverse-field Ising models~\cite{bravyi2015monte, crosson2025classical}, XY models~\cite{bravyi2017polynomial, rayudu2025fast}, and Heisenberg models on star-like bipartite graphs~\cite{takahashi2024rapidly}. This suggests that stoquastic systems may provide a suitable setting for applying fast polymer methods.

A key challenge in applying fast polymer algorithms to quantum systems is the computational cost of evaluating polymer weights and their possible negativity. For quantum systems, polymer weights are typically expressed as a trace over a Hilbert space whose dimension is exponential in the size of the polymer. For stoquastic systems, these weights admit non-negative representations that can be sampled to obtain probabilistic estimators. The challenge is therefore to combine these estimators with fast polymer sampling and counting algorithms while controlling their computational cost.

In this paper, we establish a general framework for developing fast sampling and counting algorithms for stoquastic spin systems at high temperature. Our framework is based on the polymer dynamics of Chen et al.~\cite{chen2021fast} and the graphlet sampling algorithm of Blanca et al.~\cite{blanca2024fast}. We apply our framework to obtain fast algorithms for approximating the partition function and sampling from the thermal distribution of (1) general stoquastic spin systems, (2) ferromagnetic Heisenberg models, and (3) antiferromagnetic Heisenberg models on bipartite graphs. For the Heisenberg models, we obtain an improved bound on the inverse temperature by using the cycle representation of T\'oth~\cite{toth1993improved, goldschmidt2011quantum} and the loop representation of Aizenman and Nachtergaele~\cite{aizenman1994geometric, goldschmidt2011quantum}.

This paper is structured as follows. In Section~\ref{section:Preliminaries}, we introduce the necessary preliminaries. Then, in Section~\ref{section:SamplingAlgorithms} and Section~\ref{section:CountingAlgorithms}, we establish our sampling and counting algorithms for polymer models. In Section~\ref{section:TruncatedPolymerModels}, we extend our algorithms to truncated polymer models. In Section~\ref{section:Applications}, we apply our algorithms to stoquastic spin systems. Finally, we conclude in Section~\ref{section:ConclusionOutlook} with some remarks and open problems.

\section{Preliminaries}
\label{section:Preliminaries}

\subsection{Graph Theory}
\label{section:GraphTheory}

Let $G$ be a graph with vertex set $V(G)$ and edge set $E(G)$. We denote the \emph{order} of $G$ by $\abs{G}\coloneqq\abs{V(G)}$ and the \emph{size} of $G$ by $\norm{G}\coloneqq\abs{E(G)}$. The \emph{induced subgraph} of a subset of vertices $U \subseteq V(G)$ is the graph $G[U]$ whose vertex set is $U$ and whose edge set consists of all edges in $G$ which have both endpoints in $U$. A \emph{graphlet} of $G$ is a connected induced subgraph of $G$. The \emph{maximum degree} of $G$ is the maximum degree over all vertices of $G$. We denote the \emph{open neighbourhood} of a vertex $v \in V(G)$ by $\mathcal{N}(v)\coloneqq\{u \mid \{u,v\} \in E(G)\}$. More generally, we denote the open neighbourhood of a subset of vertices $U \subseteq V(G)$ by $\mathcal{N}(U)\coloneqq\left(\bigcup_{v \in U}\mathcal{N}(v)\right){\setminus}U$. The \emph{line graph} of $G$ is the graph $L(G)$ whose vertex set is $E(G)$ and whose edge set consists of all pairs of edges of $G$ which share an endpoint.

\subsection{Polymer Models}
\label{section:PolymerModels}

An \emph{abstract polymer model} is a triple $\mathcal{P}=(\mathcal{C},w,\sim)$, where $\mathcal{C}$ is a countable set of objects called \emph{polymers}, $w:\mathcal{C}\to\mathbb{R}_{\geq0}$ is a function that assigns to each polymer $\gamma\in\mathcal{C}$ a non-negative weight $w_\gamma$, and $\sim$ is a symmetric compatibility relation such that each polymer is incompatible with itself. A set of polymers is called \emph{admissible} if the polymers in the set are all pairwise compatible and have positive weight. Note that the empty set is admissible. Let $\mathcal{G}$ denote the collection of all admissible sets of polymers from $\mathcal{C}$. The \emph{partition function} of $\mathcal{P}$ is defined by
\begin{equation}
    Z_\mathcal{P} \coloneqq \sum_{\Gamma\in\mathcal{G}}\prod_{\gamma\in\Gamma}w_\gamma. \notag
\end{equation}
The \emph{polymer distribution} $\mu_\mathcal{P}$ of $\mathcal{P}$ is the probability distribution on $\mathcal{G}$ defined by
\begin{equation}
    \mu_\mathcal{P}(\Gamma) \coloneqq \frac{1}{Z_\mathcal{P}}
    \prod_{\gamma\in\Gamma}w_\gamma. \notag
\end{equation}

Let $G$ be a graph. A \emph{graphlet polymer model} on $G$ is a polymer model $\mathcal{P}=(\mathcal{C},w,\sim)$, where $\mathcal{C}$ is a set of graphlets of $G$, and two polymers $\gamma$ and $\gamma'$ are compatible if and only if $G[V(\gamma) \cup V(\gamma')]$ is disconnected. We extend $w$ to all graphlets of $G$ by setting $w_\gamma=0$ for $\gamma\notin\mathcal{C}$. For each vertex $v \in V(G)$, let $\mathcal{C}_v\coloneqq\{\gamma\in\mathcal{C} \mid v \in V(\gamma)\}$ denote the set of graphlets that contain $v$.

We introduce the following condition on the polymer weights.
\begin{definition}[Vertex incompatibility condition]
    Fix $\theta\in(0,1)$. A graphlet polymer model $\mathcal{P}=(\mathcal{C},w,\sim)$ defined on a graph $G$ satisfies the \emph{vertex incompatibility condition} if, for all $v \in V(G)$,
    \begin{equation}
        \sum_{\gamma\in\mathcal{C}_v}\abs{\gamma}w_\gamma \leq \theta. \notag
    \end{equation}
\end{definition}

The following lemma bounds the number of graphlets of a given order containing a fixed vertex.

\begin{lemma}[{\cite[Lemma~2.1]{borgs2013left}}]
    \label{lemma:GraphletCount}
    Fix $\Delta\in\mathbb{Z}_{\geq3}$. Let $G$ be a graph of maximum degree at most $\Delta$, and let $v$ be a vertex of $G$. The number of graphlets of $G$ of order $k$ that contain $v$ is at most $(e\Delta)^{k-1}$.
\end{lemma}

The following lemma gives a sufficient condition on the polymer weights of a graphlet polymer model for the vertex incompatibility condition to hold.

\begin{lemma}
    \label{lemma:PolymerWeightVertexIncompatibilityCondition}
    Fix $\Delta\in\mathbb{Z}_{\geq3}$. Let $\mathcal{P}=(\mathcal{C},w,\sim)$ be a graphlet polymer model defined on a graph $G$ of maximum degree at most $\Delta$. Suppose that, for all $\gamma\in\mathcal{C}$, the weight $w_\gamma$ satisfies
    \begin{equation}
        w_\gamma \leq \left(\frac{1}{e^2\Delta}\right)^\abs{\gamma}. \notag
    \end{equation}
    Then $\mathcal{P}$ satisfies the vertex incompatibility condition with $\theta=\frac{1}{(e-1)^2\Delta}$.
\end{lemma}

\begin{proof}
    Fix a vertex $v \in V(G)$. The number of polymers of order $k$ in $\mathcal{C}_v$ is at most $(e\Delta)^{k-1}$ by Lemma~\ref{lemma:GraphletCount}. Therefore,
    \begin{equation}
        \sum_{\gamma\in\mathcal{C}_v}\abs{\gamma}w_\gamma \leq \sum_{k=1}^\infty k(e\Delta)^{k-1}\left(\frac{1}{e^2\Delta}\right)^k = \frac{1}{e\Delta}\sum_{k=1}^\infty ke^{-k} = \frac{1}{(e-1)^2\Delta} < 1, \notag
    \end{equation}
    completing the proof.
\end{proof}

We introduce the notion of $\kappa$-computable graphlet polymer models, which characterises the cost of evaluating polymer weights.

\begin{definition}[$\kappa$-computable graphlet polymer model]
    A graphlet polymer model $\mathcal{P}=(\mathcal{C},w,\sim)$ is $\kappa$-computable if, for all polymers $\gamma\in\mathcal{C}$, $w_\gamma$ can be computed in time $e^{\kappa\abs{\gamma}}\abs{\gamma}^{O(1)}$.
\end{definition}

Let $G$ be a graph. A \emph{subgraph polymer model} on $G$ is a polymer model $\mathcal{P}=(\mathcal{C},w,\sim)$, where $\mathcal{C}$ is a set of connected subgraphs of $G$ with at least one edge, and two polymers $\gamma$ and $\gamma'$ are compatible if and only if $V(\gamma)\cap V(\gamma')=\varnothing$. We extend $w$ to all such subgraphs by setting $w_\gamma=0$ for $\gamma\notin\mathcal{C}$. A subgraph polymer model $\mathcal{P}$ on $G$ is equivalent to a graphlet polymer model on $L(G)$ under the identification of connected subgraphs of $G$ with the corresponding connected induced subgraphs of $L(G)$.

Let $\mathcal{P}=(\mathcal{C},w,\sim)$ be a subgraph polymer model on $G$, and let $d\in\mathbb{Z}^+$. For each $\gamma\in\mathcal{C}$ with $w_\gamma>0$, let $\mu_\gamma$ be a probability distribution on $[d]^{V(\gamma)}$, and let $\mu_v$ be a probability distribution on $[d]$ for each $v \in V(G)$. The \emph{coloured polymer distribution} $\hat{\mu}_\mathcal{P}$ is the probability distribution on $[d]^{V(G)}$ defined by
\begin{equation}
    \hat{\mu}_\mathcal{P}(\sigma) \coloneqq \sum_{\Gamma\in\mathcal{G}}\mu_\mathcal{P}(\Gamma)\prod_{\gamma\in\Gamma}\mu_\gamma(\sigma_{V(\gamma)})\prod_{v \in V(G){\setminus}V(\Gamma)}\mu_v(\sigma_v), \notag
\end{equation}
where $V(\Gamma)\coloneqq\bigcup_{\gamma\in\Gamma}V(\gamma)$. The notion of $\kappa$-computability extends to subgraph polymer models: $\mathcal{P}$ is $\kappa$-computable if, for all $\gamma\in\mathcal{C}$, $w_\gamma$ can be computed in time $e^{\kappa\norm{\gamma}}\norm{\gamma}^{O(1)}$. For the coloured polymer distribution $\hat{\mu}_\mathcal{P}$, we further require that, for all $\gamma\in\mathcal{C}$, $\mu_\gamma$ can be sampled in time $e^{\kappa\norm{\gamma}}\norm{\gamma}^{O(1)}$.

\subsection{Stoquastic Spin Systems}
\label{section:StoquasticSpinSystems}

A \emph{stoquastic spin system} $\mathcal{S}$ is modelled by a graph $G$. At each vertex $v$ of $G$, there is a $d$-dimensional Hilbert space $\mathcal{H}_v$ with $d<\infty$. The Hilbert space on the graph is given by $\mathcal{H}_G\coloneqq\bigotimes_{v \in V(G)}\mathcal{H}_v$. An interaction $\Phi$ assigns a self-adjoint operator $\Phi(e)$ on $\bigotimes_{v \in e}\mathcal{H}_v$ to each edge $e$ of $G$. The Hamiltonian of $\mathcal{S}$ is defined by $H_\mathcal{S}\coloneqq\sum_{e \in E(G)}\Phi(e)$. We consider stoquastic Hamiltonians in the sense that all matrix elements of $\Phi(e)$ are non-positive in the spin basis for all $e \in E(G)$. This differs from the usual definition in that we also require the diagonal matrix elements to be non-positive. This condition can always be achieved by subtracting a suitable multiple of the identity from each $\Phi(e)$.

At inverse temperature $\beta\geq0$, the \emph{partition function} $Z_\mathcal{S}$ is defined by $Z_\mathcal{S}\coloneqq\Tr\left[e^{-\beta H_\mathcal{S}}\right]$ and the \emph{thermal state} $\rho_\mathcal{S}$ is defined by $\rho_\mathcal{S}\coloneqq(Z_\mathcal{S})^{-1}e^{-\beta H_\mathcal{S}}$. The \emph{thermal distribution} $\mu_{\rho_\mathcal{S}}$ over the spin space $[d]^{V(G)}$ is defined by $\mu_{\rho_\mathcal{S}}(\sigma)\coloneqq\Tr\left[\ketbra{\sigma}{\sigma}\rho_\mathcal{S}\right]$ for all $\sigma\in[d]^{V(G)}$. We restrict our attention to stoquastic spin systems modelled by bounded-degree graphs. We shall assume that the trace is normalised so that $\Tr[\mathbb{I}]=1$, which is equivalent to rescaling the partition function and thermal state by a multiplicative factor and does not affect the thermal distribution. Further, we assume that $0<\norm{\Phi(e)}\leq1$ for every $e \in E(G)$, where $\norm{\;\cdot\;}$ denotes the operator norm. Note that this is always possible by a rescaling of $\beta$.

\subsection{Approximation Schemes}
\label{section:ApproximationSchemes}

An \emph{$\epsilon$-approximate sampling algorithm} for a probability distribution $\mu$ is a randomised algorithm that, for any $\epsilon>0$, outputs a sample from a distribution $\hat{\mu}$ such that $\norm{\mu-\hat{\mu}}_\mathrm{TV}\leq\epsilon$. An \emph{$\epsilon$-approximate counting algorithm} for a positive real number $Z$ is a randomised algorithm that, for any $\epsilon>0$, outputs an estimate $\widehat{Z}$ such that $\abs{Z-\widehat{Z}} \leq \epsilon Z$ with probability at least $\frac{2}{3}$.

\section{Sampling Algorithms}
\label{section:SamplingAlgorithms}

In this section we establish an efficient algorithm for sampling from the polymer distribution $\mu_\mathcal{P}$ based on the polymer dynamics of Chen et al.~\cite{chen2021fast} and the graphlet sampling algorithm of Blanca et al.~\cite{blanca2024fast}.

\subsection{Polymer Dynamics}
\label{section:PolymerDynamics}

In this section we introduce the polymer dynamics, a Markov chain on the admissible sets of a graphlet polymer model whose stationary distribution is $\mu_\mathcal{P}$. This Markov chain was introduced in Ref.~\cite{chen2021fast}.

Let $\mathcal{P}=(\mathcal{C},w,\sim)$ be a graphlet polymer model defined on a graph $G$. For each vertex $v \in V(G)$, define $w_v\coloneqq\sum_{\gamma\in\mathcal{C}_v}w_\gamma$. Let $\pi_v$ denote the probability distribution on $\{\varnothing\}\cup\mathcal{C}_v$ defined by $\pi_v(\gamma)=\frac{w_\gamma}{2}$ for $\gamma\in\mathcal{C}_v$ and $\pi_v(\varnothing)=1-\frac{w_v}{2}$. Note that for polymer models satisfying the vertex incompatibility condition, we have $w_v<1$, and hence $\pi_v$ defines a valid probability distribution on $\{\varnothing\}\cup\mathcal{C}_v$.

\begin{definition}[Polymer dynamics]
    Let $\mathcal{P}=(\mathcal{C},w,\sim)$ be a graphlet polymer model defined on a graph $G$, and let $\mathcal{G}$ denote the collection of all admissible sets of $\mathcal{P}$. The \emph{polymer dynamics} on $\mathcal{P}$ is the Markov chain $(\Gamma_t)_{t\in\mathbb{N}}$ on $\mathcal{G}$ defined by the following transition from $\Gamma_t$:
    \begin{enumerate}
        \item Choose a vertex $v \in V(G)$ uniformly at random, and let
        \begin{equation}
            \gamma^* \coloneqq
            \begin{cases}
                \gamma, & \text{if there exists $\gamma\in\Gamma_t\cap\mathcal{C}_v$}, \\
                \varnothing, & \text{otherwise}.
            \end{cases}
            \notag
        \end{equation}
        \item Choose one of the following two updates:
        \begin{enumerate}
            \item With probability $\frac{1}{3}$, set
            \begin{equation}
                \Gamma_{t+1} = \Gamma_t{\setminus}\{\gamma^*\}. \notag
            \end{equation}
            \item With probability $\frac{2}{3}$, sample $\gamma\sim\pi_v$, and set
            \begin{equation}
                \Gamma_{t+1} =
                \begin{cases}
                    \Gamma_t\cup\{\gamma\}, & \text{if $\Gamma_t\cup\{\gamma\}\in\mathcal{G}$}, \\
                    \Gamma_t, & \text{otherwise}.
                \end{cases}
                \notag
            \end{equation}
        \end{enumerate}
    \end{enumerate}
\end{definition}

The following lemma shows that the polymer dynamics has unique stationary distribution $\mu_\mathcal{P}$.

\begin{lemma}
    \label{lemma:PolymerDynamicsUniqueStationaryDistribution}
    Let $\mathcal{P}$ be a graphlet polymer model satisfying the vertex incompatibility condition. The polymer dynamics on $\mathcal{P}$ has unique stationary distribution $\mu_\mathcal{P}$.
\end{lemma}

The proof of Lemma~\ref{lemma:PolymerDynamicsUniqueStationaryDistribution} follows that of Ref.~\cite{chen2021fast}.

\begin{proof}
    The polymer dynamics is irreducible since we can transition between any two admissible sets $\Gamma\in\mathcal{G}$ and $\Gamma^*\in\mathcal{G}$ via the empty set. The polymer dynamics is aperiodic since there is a positive probability that no transition occurs. Since the polymer dynamics is finite, irreducible, and aperiodic, it is ergodic, and thus admits a unique stationary distribution. We show that this distribution is $\mu_\mathcal{P}$ by showing that the polymer dynamics satisfies the detailed balance condition.

    A single update of the dynamics changes a configuration $\Gamma$ by the addition or removal of at most one polymer $\gamma^*$. Let $\Gamma^*=\Gamma\cup\{\gamma^*\}$. Then,
    \begin{equation}
        \frac{\mu_\mathcal{P}(\Gamma^*)}{\mu_\mathcal{P}(\Gamma)} = w_{\gamma^*} = \frac{\frac{2\abs{\gamma^*}}{3\abs{G}}\cdot\frac{w_{\gamma^*}}{2}}{\frac{\abs{\gamma^*}}{3\abs{G}}} = \frac{P(\Gamma,\Gamma^*)}{P(\Gamma^*,\Gamma)}, \notag
    \end{equation}
    where $P$ denotes the transition matrix of the polymer dynamics. Therefore, the detailed balance condition holds with respect to $\mu_\mathcal{P}$, and hence $\mu_\mathcal{P}$ is the unique stationary distribution. This completes the proof.
\end{proof}

The following lemma bounds the mixing time of the polymer dynamics. Recall that for a Markov chain on a state space $\mathcal{G}$ with transition matrix $P$ and unique stationary distribution $\mu$, the \emph{mixing time} is defined by $\tau_\text{mix}(\epsilon)\coloneqq\min\left\{t\in\mathbb{Z}^+\mid\max_{\Gamma\in\mathcal{G}}\norm{P^t(\Gamma,\;\cdot\;)-\mu}_\mathrm{TV}\leq\epsilon\right\}$ for $\epsilon>0$. The lemma shows that the polymer dynamics is rapidly mixing.

\begin{lemma}[{\cite[Theorem~2]{chen2021fast}}]
    \label{lemma:PolymerDynamicsMixingTime}
    Let $\mathcal{P}$ be a graphlet polymer model defined on a graph $G$ satisfying the vertex incompatibility condition. The polymer dynamics on $\mathcal{P}$ has mixing time $\tau_\text{mix}(\epsilon)=O(\abs{G}\log(\abs{G}/\epsilon))$.
\end{lemma}

It remains to show that each step of the polymer dynamics can be implemented efficiently.

\subsection{Graphlet Sampler}
\label{section:GraphletSampler}

In this section we introduce the graphlet sampler, an algorithm for sampling from an auxiliary distribution $\hat{\pi}_v$ for a given vertex $v \in V(G)$, which will be used as a subroutine to sample from $\pi_v$. The graphlet sampler is a variant of that of Ref.~\cite{blanca2024fast}, adapted to allow for $\kappa$-computable graphlet polymer models and runtime bounds depending on the order of the largest polymer.

Let $\mathcal{P}=(\mathcal{C},w,\sim)$ be a graphlet polymer model defined on a graph $G$. For each vertex $v \in V(G)$, define $\widehat{Z}_v\coloneqq\sum_{\gamma\in\{\varnothing\}\cup\mathcal{C}_v}w_\gamma$, where $w_\varnothing\coloneqq1$. Let $\hat{\pi}_v$ denote the probability distribution on $\{\varnothing\}\cup\mathcal{C}_v$ defined by $\hat{\pi}_v(\gamma)=\frac{w_\gamma}{\widehat{Z}_v}$. For a graphlet $\gamma$, we denote its \emph{boundary} by $\partial\gamma\coloneqq\mathcal{N}(V(\gamma))$, with the convention $\partial\varnothing\coloneqq\{v\}$.

\begin{definition}[Graphlet sampler]
    Fix $\Delta\in\mathbb{Z}_{\geq3}$ and $p,w_\star\in(0,1)$. Let $\mathcal{P}=(\mathcal{C},w,\sim)$ be a graphlet polymer model defined on a graph $G$ of maximum degree at most $\Delta$, and let $v$ be a vertex of $G$. Let $m=\max_{\gamma\in\mathcal{C}_v}\abs{\gamma}$ denote the order of the largest polymer in $\mathcal{C}_v$. Suppose that, for all $\gamma\in\mathcal{C}_v$, the weight $w_\gamma$ satisfies $w_\gamma \leq w_\star^\abs{\gamma}$. The \emph{graphlet sampler} is the following randomised algorithm:
    \begin{enumerate}
        \item Initialise a queue $Q=\varnothing$ and set $S=\varnothing$.
        \item Choose one of the following:
         \begin{enumerate}
            \item With probability $p$, add $v$ to $Q$ and $S$, and set $v$ as processed.
            \item With probability $1-p$, set $v$ as processed.
        \end{enumerate}
        \item While $Q\neq\varnothing$:
        \begin{enumerate}
            \item Remove a vertex $u$ from $Q$.
            \item For each unprocessed neighbour $w$ of $u$, choose one of the following:
            \begin{enumerate}
                \item With probability $p$, add $w$ to $Q$ and $S$, and set $w$ as processed.
                \item With probability $1-p$, set $w$ as processed.
            \end{enumerate}
            \item If the number of vertices in $S$ is greater than $m$, restart the algorithm.
        \end{enumerate}
        \item Let $\gamma$ be the graphlet with vertex set $S$. Choose one of the following:
        \begin{enumerate}
            \item With probability $w_\gamma w_\star^{-\abs{\gamma}}(1-p)^{(\Delta-2)\abs{\gamma}-\abs{\partial\gamma}+2}$, return $\gamma$.
            \item With probability $1-w_\gamma w_\star^{-\abs{\gamma}}(1-p)^{(\Delta-2)\abs{\gamma}-\abs{\partial\gamma}+2}$, restart the algorithm.
        \end{enumerate}
    \end{enumerate}
\end{definition}

The following lemma establishes that the output distribution of the graphlet sampler is $\hat{\pi}_v$ and bounds the expected runtime.

\begin{lemma}
    \label{lemma:GraphletSampler}
    Fix $\Delta\in\mathbb{Z}_{\geq3}$ and $r,\kappa\geq0$. Let $\mathcal{P}=(\mathcal{C},w,\sim)$ be a graphlet polymer model defined on a graph $G$ of maximum degree at most $\Delta$, and let $v$ be a vertex of $G$. Let $m=\max_{\gamma\in\mathcal{C}_v}\abs{\gamma}$ denote the order of the largest polymer in $\mathcal{C}_v$. Suppose that, for all $\gamma\in\mathcal{C}_v$, the weight $w_\gamma$ satisfies
    \begin{equation}
        w_\gamma \leq \left(\frac{1}{e^{r+2}\Delta}\right)^\abs{\gamma}. \notag
    \end{equation}
    Then the output distribution of the graphlet sampler is $\hat{\pi}_v$. Suppose further that $\mathcal{P}$ is $\kappa$-computable. Then the graphlet sampler has expected runtime $O(1)$ for $r\geq\kappa$ and expected runtime $m^{O(1)}e^{(\kappa-r)m}$ for $r<\kappa$.
\end{lemma}

The proof of Lemma~\ref{lemma:GraphletSampler} follows that of Ref.~\cite{blanca2024fast}, with modifications to allow for $\kappa$-computable graphlet polymer models and runtime bounds depending on the order of the largest polymer.

\begin{proof}
    Let $w_\star\coloneqq\frac{1}{e^{r+2}\Delta}$, and let $p\in\left[0,\frac{1}{\Delta-1}\right]$ be the unique solution to $f(p)=w_\star$, where $f(x) \coloneqq x(1-x)^{\Delta-2}$, which exists by monotonicity of $f$ on $\left[0,\frac{1}{\Delta-1}\right]$.
    
    We first show that the output distribution of the graphlet sampler is $\hat{\pi}_v$. The probability of outputting $\gamma\in\{\varnothing\}\cup\mathcal{C}_v$ in an iteration of the algorithm is
    \begin{align}
        p^\abs{\gamma}(1-p)^\abs{\partial\gamma}w_\gamma w_\star^{-\abs{\gamma}}(1-p)^{(\Delta-2)\abs{\gamma}-\abs{\partial\gamma}+2} &= \left(p(1-p)^{\Delta-2}\right)^\abs{\gamma}w_\star^{-\abs{\gamma}}(1-p)^2w_\gamma \notag \\
        &= (1-p)^2w_\gamma, \notag
    \end{align}
    where we have used $p(1-p)^{\Delta-2}=w_\star$. The probability of acceptance per iteration is then
    \begin{equation}
        (1-p)^2\sum_{\gamma\in\{\varnothing\}\cup\mathcal{C}_v}w_\gamma = (1-p)^2\widehat{Z}_v. \notag
    \end{equation}
    Therefore, the probability that $\gamma\in\{\varnothing\}\cup\mathcal{C}_v$ is the output of the graphlet sampler is
    \begin{equation}
        \sum_{t=1}^\infty\left(1-(1-p)^2\widehat{Z}_v\right)^{t-1}(1-p)^2w_\gamma = \frac{w_\gamma}{\widehat{Z}_v} = \hat{\pi}_v(\gamma). \notag
    \end{equation}
    Hence, the graphlet sampler outputs a sample from $\hat{\pi}_v$.

    We next bound the expected runtime of the graphlet sampler. In each iteration, the breadth-first search percolation produces $\gamma$ in time $O(\abs{\gamma})$. Since $\mathcal{P}$ is $\kappa$-computable, each iteration takes time $e^{\kappa\abs{\gamma}}\abs{\gamma}^{O(1)}$ when $\abs{\gamma} \leq m$, and time $\abs{\gamma}^{O(1)}$ when $\abs{\gamma} > m$, as the polymer is rejected without computing its weight. Let $\nu_v$ denote the output probability distribution of the percolation step. The expected runtime per iteration is therefore
    \begin{equation}
        \mathbb{E}_{\nu_v}\left[\left(e^{\kappa\abs{\gamma}}\mathds{1}_{\abs{\gamma} \leq m}+\mathds{1}_{\abs{\gamma} > m}\right)\abs{\gamma}^{O(1)}\right]. \notag
    \end{equation}

    We now show that the order $\abs{\gamma}$ of $\gamma$ under $\nu_v$ is stochastically dominated by the total progeny of a Galton--Watson process with offspring distribution $\operatorname{Poi}(\lambda)$, where $\lambda\coloneqq(\Delta+1)p$. Recall that a Galton--Watson process is a branching process initiated from a single particle, where each particle independently produces a random number of offspring according to a fixed distribution; the total progeny $Z$ is the total number of particles. Let $\nu_\lambda$ denote the probability distribution of $Z$ under this process.

    The breadth-first search percolation includes each of at most $\Delta$ neighbours independently with probability $p$. Therefore, $\abs{\gamma}$ is stochastically dominated by the total progeny of a Galton--Watson process with offspring distribution $\operatorname{Bin}(\Delta,p)$. Since $\operatorname{Bin}(\Delta,p)$ is stochastically dominated by $\operatorname{Poi}(\lambda)$, $\abs{\gamma}$ is stochastically dominated by $Z\sim\nu_\lambda$.

    By monotonicity of $f(x)=x(1-x)^{\Delta-2}$ on $\left[0,\frac{1}{\Delta-1}\right]$, we have that $e^{-(r+2)} \leq \lambda \leq e^{-(r+1)}$, and hence the Galton--Watson process is subcritical, and $Z$ is almost surely finite. For a subcritical Galton--Watson process with offspring distribution $\operatorname{Poi}(\lambda)$, the total progeny satisfies (see for example Ref.~\cite{roch2023modern})
    \begin{equation}
        \Pr_{\nu_\lambda}\left[Z \geq k\right] \leq \lambda^{k-1}e^{(1-\lambda)k} \leq e^{r+1}e^{-\left(r+e^{-(r+2)}\right)k}. \notag
    \end{equation}
    Therefore, the expected runtime per iteration satisfies
    \begin{align}
        \mathbb{E}_{\nu_v}\left[\left(e^{\kappa\abs{\gamma}}\mathds{1}_{\abs{\gamma} \leq m}+\mathds{1}_{\abs{\gamma} > m}\right)\abs{\gamma}^{O(1)}\right] &\leq \mathbb{E}_{\nu_v}\left[e^{\kappa\min(\abs{\gamma},m)}\abs{\gamma}^{O(1)}\right] \notag \\
        &\leq \mathbb{E}_{\nu_\lambda}\left[e^{\kappa\min(Z,m)}Z^{O(1)}\right] \notag \\
        &\leq \sum_{k=1}^\infty e^{\kappa\min(k,m)}k^{O(1)}\Pr_{\nu_\lambda}\left[Z \geq k\right] \notag \\
        &\leq e^{r+1}\sum_{k=1}^\infty e^{\kappa\min(k,m)}k^{O(1)}e^{-\left(r+e^{-(r+2)}\right)k} \notag \\
        &= 
        \begin{cases}
            O(1), & r\geq\kappa, \\
            m^{O(1)}e^{(\kappa-r)m}, & r<\kappa.
        \end{cases}
        \notag
    \end{align}
    The total expected runtime of the graphlet sampler is then
    \begin{equation}
        \frac{\mathbb{E}_{\nu_v}\left[\left(e^{\kappa\abs{\gamma}}\mathds{1}_{\abs{\gamma} \leq m}+\mathds{1}_{\abs{\gamma} > m}\right)\abs{\gamma}^{O(1)}\right]}{(1-p)^2\widehat{Z}_v} = 
        \begin{cases}
            O(1), & r\geq\kappa, \\
            m^{O(1)}e^{(\kappa-r)m}, & r<\kappa.
        \end{cases}
        \notag
    \end{equation}
    This completes the proof.
\end{proof}

We now introduce probabilistic polymer models and show that it suffices to compute polymer weights via an unbiased estimator.

Let $\mathcal{P}=(\mathcal{C},w,\sim)$ be a polymer model. A \emph{probabilistic polymer model} of $\mathcal{P}$ is a polymer model $\mathcal{P}_\circ=(\mathcal{C},\tilde{w},\sim)$ where, for each $\gamma\in\mathcal{C}$, the weight $\tilde{w}_\gamma$ is a non-negative random variable satisfying $\mathbb{E}\left[\tilde{w}_\gamma\right]=w_\gamma$. A probabilistic polymer model $\mathcal{P}_\circ$ is $\kappa$-computable if, for all polymers $\gamma\in\mathcal{C}$, a sample of $\tilde{w}_\gamma$ can be obtained in expected time $e^{\kappa\abs{\gamma}}\abs{\gamma}^{O(1)}$.

The following lemma shows that the graphlet sampler extends to probabilistic polymer models with output distribution $\hat{\pi}_v$, provided the estimator is bounded by the weight decay condition.

\begin{lemma}
    \label{lemma:ProbabilisticGraphletSampler}
    Fix $\Delta\in\mathbb{Z}_{\geq3}$ and $r,\kappa\geq0$. Let $\mathcal{P}_\circ=(\mathcal{C},\tilde{w},\sim)$ be a probabilistic polymer model of a graphlet polymer model $\mathcal{P}=(\mathcal{C},w,\sim)$ defined on a graph $G$ of maximum degree at most $\Delta$, and let $v$ be a vertex of $G$. Let $m=\max_{\gamma\in\mathcal{C}_v}\abs{\gamma}$ denote the order of the largest polymer in $\mathcal{C}_v$. Suppose that, for all $\gamma\in\mathcal{C}_v$,
    \begin{equation}
        \tilde{w}_\gamma \leq \left(\frac{1}{e^{r+2}\Delta}\right)^\abs{\gamma}. \notag
    \end{equation}
    Then the output distribution of the graphlet sampler using weights $\tilde{w}_\gamma$ is $\hat{\pi}_v$. Suppose further that $\mathcal{P}_\circ$ is $\kappa$-computable. Then the graphlet sampler has expected runtime $O(1)$ for $r\geq\kappa$ and expected runtime $m^{O(1)}e^{(\kappa-r)m}$ for $r<\kappa$.
\end{lemma}

\begin{proof}
    The proof follows that of Lemma~\ref{lemma:GraphletSampler}, with the observation that the probability of outputting $\gamma\in\{\varnothing\}\cup\mathcal{C}_v$ in an iteration of the algorithm is $\mathbb{E}\left[(1-p)^2\tilde{w}_\gamma\right]=(1-p)^2w_\gamma$.
\end{proof}

\subsection{Single Polymer Sampler}
\label{section:SinglePolymerSampler}

In this section we introduce the single polymer sampler, a procedure for sampling from $\pi_v$ given a vertex $v \in V(G)$, using the graphlet sampler as a subroutine. The single polymer sampler is a variant of that of Ref.~\cite{blanca2024fast}.

\begin{definition}[Single polymer sampler]
    Let $\mathcal{P}=(\mathcal{C},w,\sim)$ be a graphlet polymer model defined on a graph $G$, and let $v$ be a vertex of $G$. The \emph{single polymer sampler} is the following randomised algorithm:
    \begin{enumerate}
        \item Run the graphlet sampler to obtain a sample $\gamma\sim\hat{\pi}_v$.
        \item Using an independent iteration of the graphlet sampler, obtain a sample $X\sim\operatorname{Ber}\left(\frac{\widehat{Z}_v}{2}\right)$.
        \item Choose one of the following:
        \begin{enumerate}
            \item If $X=1$, return $\gamma$.
            \item If $X=0$, return $\varnothing$.
        \end{enumerate}
    \end{enumerate}
\end{definition}

The following lemma establishes that the output distribution of the single polymer sampler is $\pi_v$ and bounds the expected runtime.

\begin{lemma}
    \label{lemma:SinglePolymerSampler}
    Fix $\Delta\in\mathbb{Z}_{\geq3}$ and $r,\kappa\geq0$. Let $\mathcal{P}=(\mathcal{C},w,\sim)$ be a graphlet polymer model defined on a graph $G$ of maximum degree at most $\Delta$, and let $v$ be a vertex of $G$. Let $m=\max_{\gamma\in\mathcal{C}_v}\abs{\gamma}$ denote the order of the largest polymer in $\mathcal{C}_v$. Suppose that, for all $\gamma\in\mathcal{C}_v$, the weight $w_\gamma$ satisfies
    \begin{equation}
        w_\gamma \leq \left(\frac{1}{e^{r+2}\Delta}\right)^\abs{\gamma}. \notag
    \end{equation}
    Then the output distribution of the single polymer sampler is $\pi_v$. Suppose further that $\mathcal{P}$ is $\kappa$-computable. Then the single polymer sampler has expected runtime $O(1)$ for $r\geq\kappa$ and expected runtime $m^{O(1)}e^{(\kappa-r)m}$ for $r<\kappa$.
\end{lemma}

The proof of Lemma~\ref{lemma:SinglePolymerSampler} follows that of Ref.~\cite{blanca2024fast}.

\begin{proof}
    We first show that the output distribution of the single polymer sampler is $\pi_v$. Since $\mathcal{P}$ satisfies the vertex incompatibility condition by Lemma~\ref{lemma:PolymerWeightVertexIncompatibilityCondition}, we have 
    \begin{equation}
        \widehat{Z}_v = 1+\sum_{\gamma\in\mathcal{C}_v}w_\gamma \leq 1+\sum_{\gamma\in\mathcal{C}_v}\abs{\gamma}w_\gamma < 2, \notag
    \end{equation}
    and hence $\frac{\widehat{Z}_v}{2}$ is a valid probability. By Lemma~\ref{lemma:GraphletSampler}, the graphlet sampler outputs a sample from $\hat{\pi}_v$. The probability of outputting $\gamma\in\mathcal{C}_v$ is therefore
    \begin{equation}
        \hat{\pi}_v(\gamma)\frac{\widehat{Z}_v}{2} = \frac{w_\gamma}{\widehat{Z}_v}\cdot\frac{\widehat{Z}_v}{2} = \frac{w_\gamma}{2} = \pi_v(\gamma), \notag
    \end{equation}
    and the probability of outputting $\varnothing$ is
    \begin{equation}
        1-\sum_{\gamma\in\mathcal{C}_v}\frac{w_\gamma}{2} = 1-\frac{w_v}{2} = \pi_v(\varnothing). \notag
    \end{equation}
    Hence, the single polymer sampler outputs a sample from $\pi_v$.

    We next show that a sample from $\operatorname{Ber}\left(\frac{\widehat{Z}_v}{2}\right)$ can be obtained from a single iteration of the graphlet sampler. Recall from the proof of Lemma~\ref{lemma:GraphletSampler} that an iteration of the graphlet sampler accepts with probability $(1-p)^2\widehat{Z}_v$, and that $p\leq\frac{1}{e(\Delta+1)}$. Observe that we can sample from $\operatorname{Ber}\left(\frac{\widehat{Z}_v}{2}\right)$ by running a single iteration of the graphlet sampler and outputting $1$ with probability $\frac{1}{2(1-p)^2}$ if the iteration accepts and $0$ otherwise. Indeed, $\frac{1}{2(1-p)^2}$ is a valid probability, and the probability of outputting $1$ is $\frac{\widehat{Z}_v}{2}$.

    We now bound the expected runtime of the single polymer sampler. By Lemma~\ref{lemma:GraphletSampler}, the graphlet sampler has expected runtime $O(1)$ for $r\geq\kappa$ and $m^{O(1)}e^{(\kappa-r)m}$ for $r<\kappa$. Therefore, the total expected runtime of the single polymer sampler is $O(1)$ for $r\geq\kappa$ and $m^{O(1)}e^{(\kappa-r)m}$ for $r<\kappa$, completing the proof.
\end{proof}

\begin{remark}
    The single polymer sampler outputs a perfect sample from $\pi_v$, and so can be used as a subroutine in the polymer dynamics bounding chain of Ref.~\cite{blanca2024fast} to obtain a perfect sample from $\mu_\mathcal{P}$ via coupling from the past.
\end{remark}

\subsection{Polymer Sampler}
\label{section:PolymerSampler}

We now combine the polymer dynamics with the single polymer sampler to obtain our main sampling result.

\begin{theorem}
    \label{theorem:PolymerSamplingAlgorithm}
    Fix $\Delta\in\mathbb{Z}_{\geq3}$ and $r,\kappa\geq0$. Let $\mathcal{P}=(\mathcal{C},w,\sim)$ be a graphlet polymer model defined on a graph $G$ of maximum degree at most $\Delta$. Let $m=\max_{\gamma\in\mathcal{C}}\abs{\gamma}$ denote the order of the largest polymer in $\mathcal{C}$. Suppose that, for all $\gamma\in\mathcal{C}$, the weight $w_\gamma$ satisfies
    \begin{equation}
        w_\gamma \leq \left(\frac{1}{e^{r+2}\Delta}\right)^\abs{\gamma}. \notag
    \end{equation}
    Suppose further that $\mathcal{P}$ is $\kappa$-computable. Then, for any $\epsilon>0$, there is an $\epsilon$-approximate sampling algorithm for $\mu_\mathcal{P}$ with expected runtime $O(\abs{G}\log(\abs{G}/\epsilon))$ for $r\geq\kappa$ and expected runtime $\abs{G}\log(\abs{G}/\epsilon)m^{O(1)}e^{(\kappa-r)m}$ for $r<\kappa$.
\end{theorem}

\begin{proof}
    Since $\mathcal{P}$ satisfies the vertex incompatibility condition by Lemma~\ref{lemma:PolymerWeightVertexIncompatibilityCondition}, the polymer dynamics on $\mathcal{P}$ has unique stationary distribution $\mu_\mathcal{P}$ by Lemma~\ref{lemma:PolymerDynamicsUniqueStationaryDistribution} and mixing time $\tau_\text{mix}(\epsilon)=O(\abs{G}\log(\abs{G}/\epsilon))$ by Lemma~\ref{lemma:PolymerDynamicsMixingTime}. Running the polymer dynamics for $\tau_\text{mix}(\epsilon)$ steps therefore produces a sample from a distribution that is within $\epsilon$ of $\mu_\mathcal{P}$ in total variation distance. Each step of the polymer dynamics requires a sample from $\pi_v$ for a uniformly random vertex $v \in V(G)$, which is provided by the single polymer sampler in expected time $O(1)$ for $r\geq\kappa$ and $m^{O(1)}e^{(\kappa-r)m}$ for $r<\kappa$ by Lemma~\ref{lemma:SinglePolymerSampler}. The claimed runtime then follows, completing the proof.
\end{proof}

\begin{remark}
    Theorem~\ref{theorem:PolymerSamplingAlgorithm} holds when the polymer weights are given by an unbiased estimator satisfying the conditions of Lemma~\ref{lemma:ProbabilisticGraphletSampler}.
\end{remark}

\section{Counting Algorithms}
\label{section:CountingAlgorithms}

In this section we establish an efficient algorithm for approximating the partition function $Z_\mathcal{P}$ based on the simulated annealing approach of Ref.~\cite{chen2021fast} and the polymer sampling algorithm. Our algorithm allows for $\kappa$-computable graphlet polymer models and runtime bounds depending on the order of the largest polymer. Our main counting result is as follows.

\begin{theorem}
    \label{theorem:PolymerCountingAlgorithm}
    Fix $\Delta\in\mathbb{Z}_{\geq3}$ and $r,\kappa\geq0$. Let $\mathcal{P}=(\mathcal{C},w,\sim)$ be a graphlet polymer model defined on a graph $G$ of maximum degree at most $\Delta$. Let $m=\max_{\gamma\in\mathcal{C}}\abs{\gamma}$ denote the order of the largest polymer in $\mathcal{C}$. Suppose that, for all $\gamma\in\mathcal{C}$, the weight $w_\gamma$ satisfies
    \begin{equation}
        w_\gamma \leq \left(\frac{1}{e^{r+2}\Delta}\right)^\abs{\gamma}. \notag
    \end{equation}
    Suppose further that $\mathcal{P}$ is $\kappa$-computable. Then, for any $\epsilon>0$, there is an $\epsilon$-approximate counting algorithm for $Z_\mathcal{P}$ with expected runtime $O(\abs{G}^2\epsilon^{-2}\log(\abs{G}/\epsilon)^2)$ for $r\geq\kappa$ and expected runtime $\abs{G}^2\epsilon^{-2}\log(\abs{G}/\epsilon)^2m^{O(1)}e^{(\kappa-r)m}$ for $r<\kappa$.
\end{theorem}

\begin{proof}
    The proof follows by the simulated annealing approach of Ref.~\cite[Theorem~6]{chen2021fast} with the sampling oracle provided by Theorem~\ref{theorem:PolymerSamplingAlgorithm}.
\end{proof}

\section{Truncated Polymer Models}
\label{section:TruncatedPolymerModels}

In this section we introduce truncated polymer models and show that truncating polymers to a bounded order allows us to apply the polymer sampling and counting algorithms with runtime depending on the truncation order rather than the order of the largest polymer. This comes at the cost of an approximation error in the total variation distance and the partition function.

Let $\mathcal{P}=(\mathcal{C},w,\sim)$ be a graphlet polymer model. For $m\in\mathbb{Z}^+$, the \emph{truncated polymer model} $\mathcal{P}_m$ is defined by $\mathcal{P}_m\coloneqq(\mathcal{C}_m,w,\sim)$, where $\mathcal{C}_m\coloneqq\{\gamma\in\mathcal{C} \mid \abs{\gamma} \leq m\}$. Let $\mathcal{G}_m$ denote the collection of all admissible sets of $\mathcal{P}_m$.

The following lemma shows that, for sufficiently large $m$, $Z_{\mathcal{P}_m}$ is a multiplicative approximation to $Z_\mathcal{P}$ and that $\mu_{\mathcal{P}_m}$ is close to $\mu_\mathcal{P}$ in total variation distance. 

\begin{lemma}
    \label{lemma:TruncatedPolymerModelApproximation}
    Fix $\Delta\in\mathbb{Z}_{\geq3}$ and $\epsilon>0$. Let $\mathcal{P}=(\mathcal{C},w,\sim)$ be a graphlet polymer model defined on a graph $G$ of maximum degree at most $\Delta$. Suppose that, for all $\gamma\in\mathcal{C}$, the weight $w_\gamma$ satisfies
    \begin{equation}
        w_\gamma \leq \left(\frac{1}{e^2\Delta}\right)^\abs{\gamma}. \notag
    \end{equation}
    Then, for $m=O(\log(\abs{G}/\epsilon))$, the truncated polymer model $\mathcal{P}_m$ satisfies $\abs{Z_\mathcal{P}-Z_{\mathcal{P}_m}} \leq \epsilon Z_\mathcal{P}$ and $\norm{\mu_\mathcal{P}-\mu_{\mathcal{P}_m}}_\mathrm{TV}\leq\epsilon$.
\end{lemma}

\begin{proof}
    We first show that $Z_{\mathcal{P}_m}$ is a multiplicative approximation to $Z_\mathcal{P}$. We have
    \begin{equation}
        \abs{Z_\mathcal{P}-Z_{\mathcal{P}_m}} = \sum_{\Gamma\in\mathcal{G}{\setminus}\mathcal{G}_m}\prod_{\gamma\in\Gamma}w_\gamma \leq Z_\mathcal{P}\sum_{k=m+1}^\infty\sum_{\substack{\gamma\in\mathcal{C} \\ \abs{\gamma}=k}}w_\gamma. \notag
    \end{equation}
    The number of polymers of order $k$ is at most $\abs{G}(e\Delta)^{k-1}$ by Lemma~\ref{lemma:GraphletCount}. Therefore,
    \begin{equation}
        \abs{Z_\mathcal{P}-Z_{\mathcal{P}_m}} \leq Z_\mathcal{P}\abs{G}\sum_{k=m+1}^\infty(e\Delta)^{k-1}\left(\frac{1}{e^2\Delta}\right)^k =  Z_\mathcal{P}\frac{\abs{G}}{e\Delta}\sum_{k=m+1}^\infty e^{-k} \leq \epsilon Z_\mathcal{P}, \notag
    \end{equation}
    for $m=O(\log(\abs{G}/\epsilon))$. We now bound the total variation distance between $\mu_\mathcal{P}$ and $\mu_{\mathcal{P}_m}$. We have
    \begin{align}
        \norm{\mu_\mathcal{P}-\mu_{\mathcal{P}_m}}_\mathrm{TV} &= \frac{1}{2}\sum_{\Gamma\in\mathcal{G}_m}\left(\mu_{\mathcal{P}_m}(\Gamma)-\mu_\mathcal{P}(\Gamma)\right)+\frac{1}{2}\sum_{\Gamma\in\mathcal{G}{\setminus}\mathcal{G}_m}\mu_\mathcal{P}(\Gamma) \notag \\
        &= \frac{1}{2}\sum_{\Gamma\in\mathcal{G}_m}\mu_{\mathcal{P}_m}(\Gamma)\left(\frac{Z_\mathcal{P}-Z_{\mathcal{P}_m}}{Z_\mathcal{P}}\right)+\frac{1}{2}\left(\frac{Z_\mathcal{P}-Z_{\mathcal{P}_m}}{Z_\mathcal{P}}\right) \notag \\
        &= \frac{Z_\mathcal{P}-Z_{\mathcal{P}_m}}{Z_\mathcal{P}} \notag \\
        &\leq \epsilon, \notag
    \end{align}
    completing the proof.
\end{proof}

The following corollary is a consequence of Theorem~\ref{theorem:PolymerSamplingAlgorithm}, Theorem~\ref{theorem:PolymerCountingAlgorithm}, and Lemma~\ref{lemma:TruncatedPolymerModelApproximation}.

\begin{corollary}
    \label{corollary:TruncatedPolymerAlgorithm}
    Fix $\Delta\in\mathbb{Z}_{\geq3}$ and $r,\kappa\geq0$. Let $\mathcal{P}=(\mathcal{C},w,\sim)$ be a graphlet polymer model defined on a graph $G$ of maximum degree at most $\Delta$. Suppose that, for all $\gamma\in\mathcal{C}$, the weight $w_\gamma$ satisfies
    \begin{equation}
        w_\gamma \leq \left(\frac{1}{e^{r+2}\Delta}\right)^\abs{\gamma}. \notag
    \end{equation}
    Suppose further that $\mathcal{P}$ is $\kappa$-computable. Then, for any $\epsilon>0$, there is an $\epsilon$-approximate sampling algorithm for $\mu_\mathcal{P}$ and an $\epsilon$-approximate counting algorithm for $Z_\mathcal{P}$. The sampling algorithm has expected runtime $O(\abs{G}\log(\abs{G}/\epsilon))$ for $r\geq\kappa$ and expected runtime $\abs{G}^{\kappa-r+1}\epsilon^{-(\kappa-r)}\log(\abs{G}/\epsilon)^{O(1)}$ for $r<\kappa$. The counting algorithm has expected runtime $O(\abs{G}^2\epsilon^{-2}\log(\abs{G}/\epsilon)^2)$ for $r\geq\kappa$ and expected runtime $\abs{G}^{\kappa-r+2}\epsilon^{-(\kappa-r+2)}\log(\abs{G}/\epsilon)^{O(1)}$ for $r<\kappa$.
\end{corollary}

\begin{proof}
    We apply Theorem~\ref{theorem:PolymerSamplingAlgorithm} and Theorem~\ref{theorem:PolymerCountingAlgorithm} to the truncated polymer model $\mathcal{P}_m$ with $m=O(\log(\abs{G}/\epsilon))$. The weight decay condition and $\kappa$-computability are inherited from $\mathcal{P}$. By Lemma~\ref{lemma:TruncatedPolymerModelApproximation}, we have $\norm{\mu_\mathcal{P}-\mu_{\mathcal{P}_m}}_\mathrm{TV}\leq\frac{\epsilon}{2}$ and $\abs{Z_\mathcal{P}-Z_{\mathcal{P}_m}}\leq\frac{\epsilon}{2}Z_\mathcal{P}$, and so by the triangle inequality there is an $\epsilon$-approximate sampling algorithm for $\mu_\mathcal{P}$ and an $\epsilon$-approximate counting algorithm for $Z_\mathcal{P}$, completing the proof.
\end{proof}

\section{Applications}
\label{section:Applications}

In this section we apply the polymer sampling and counting algorithms to establish efficient sampling and counting algorithms for (1) general stoquastic spin systems, (2) ferromagnetic Heisenberg models, and (3) antiferromagnetic Heisenberg models on bipartite graphs.

\subsection{Stoquastic Models}
\label{section:StoquasticModels}

We first consider general stoquastic Hamiltonians. Let $\mathcal{S}=(\Phi,d,\beta)$ denote a stoquastic spin system on a graph $G$ with interaction $\Phi$, local dimension $d<\infty$, and inverse temperature $\beta\geq0$. Our first result is as follows.

\begin{theorem}
    \label{theorem:StoquasticModelAlgorithm}
    Fix $\Delta\in\mathbb{Z}_{\geq3}$. Let $\mathcal{S}=(\Phi,d,\beta)$ be a stoquastic spin system on a graph $G$ of maximum degree at most $\Delta$. Let $\beta$ be such that
    \begin{equation}
        \beta \leq \frac{1}{2e^3d^3\Delta}. \notag
    \end{equation}
    Then, for any $\epsilon>0$, there is an $\epsilon$-approximate sampling algorithm for $\mu_{\rho_\mathcal{S}}$ with expected runtime $O(\abs{G}\log(\abs{G}/\epsilon))$ and an $\epsilon$-approximate counting algorithm for $Z_\mathcal{S}$ with expected runtime $O(\abs{G}^2\epsilon^{-2}\log(\abs{G}/\epsilon)^2)$.
\end{theorem}

We prove Theorem~\ref{theorem:StoquasticModelAlgorithm} by showing that the conditions required to apply Theorem~\ref{theorem:PolymerSamplingAlgorithm} and Theorem~\ref{theorem:PolymerCountingAlgorithm} are satisfied. In particular, we show that (1) the stoquastic spin system admits a suitable subgraph polymer model representation, (2) the polymer weights satisfy the weight decay condition, and (3) the stoquastic polymer model is $\kappa$-computable.

Let $\mathcal{S}=(\Phi,d,\beta)$ be a stoquastic spin system on a graph $G$. We define the \emph{stoquastic polymer model} of $\mathcal{S}$ to be the subgraph polymer model $\mathcal{P}=(\mathcal{C},w,\sim)$ on $G$, where $\mathcal{C}$ is the set of all connected subgraphs of $G$ with at least one edge, and $w:\mathcal{C}\to\mathbb{R}_{\geq0}$ is defined by
\begin{equation}
    w_\gamma \coloneqq (-1)^\norm{\gamma}\sum_{T \subseteq E(\gamma)}(-1)^\abs{T}\Tr\left[e^{-\beta\sum_{e \in T}\Phi(e)}\right]. \notag
\end{equation}
For each $\gamma\in\mathcal{C}$ with $w_\gamma>0$, let $\mu_\gamma$ be the probability distribution on $[d]^{V(\gamma)}$ defined by
\begin{equation}
    \mu_\gamma(\sigma) \coloneqq \frac{(-1)^\norm{\gamma}}{w_\gamma}\sum_{T \subseteq E(\gamma)}(-1)^\abs{T}\Tr\left[\ketbra{\sigma}{\sigma}e^{-\beta\sum_{e \in T}\Phi(e)}\right]. \notag
\end{equation}
For each $v \in V(G)$, let $\mu_v$ be the probability distribution on $[d]$ defined by $\mu_v(\sigma)=\frac{1}{d}$ for all $\sigma\in[d]$.

\begin{lemma}
    \label{lemma:StoquasticPolymerRepresentation}
    Let $\mathcal{S}=(\Phi,d,\beta)$ be a stoquastic spin system on a graph $G$, and let $\mathcal{P}=(\mathcal{C},w,\sim)$ be the stoquastic polymer model of $\mathcal{S}$. Then the partition function of $\mathcal{S}$ satisfies $Z_\mathcal{S}=Z_\mathcal{P}$ and the thermal distribution of $\mathcal{S}$ satisfies $\mu_{\rho_\mathcal{S}}=\hat{\mu}_\mathcal{P}$.
\end{lemma}

\begin{proof}
    We first show that $Z_\mathcal{S}=Z_\mathcal{P}$, following the proof of Ref.~\cite{mann2024algorithmic}. By the principle of inclusion-exclusion (see for example Ref.~\cite[Theorem~12.1]{graham1995handbook}),
    \begin{align}
        Z_\mathcal{S} &= \Tr\left[e^{-\beta H_\mathcal{S}}\right] \notag \\
        &= \sum_{S \subseteq E(G)}(-1)^\abs{S}\sum_{T \subseteq S}(-1)^\abs{T}\Tr\left[e^{-\beta\sum_{e \in T}\Phi(e)}\right]. \notag
    \end{align}
    For a subset $S \subseteq E(G)$, let $\Gamma_S$ denote the maximally connected components of $S$. By factorising over these components,
    \begin{align}
        Z_\mathcal{S} &= \sum_{S \subseteq E(G)}\prod_{\gamma\in\Gamma_S}(-1)^\norm{\gamma}\sum_{T \subseteq E(\gamma)}(-1)^\abs{T}\Tr\left[e^{-\beta\sum_{e \in T}\Phi(e)}\right] \notag \\
        &= \sum_{S \subseteq E(G)}\prod_{\gamma\in\Gamma_S}w_\gamma \notag \\
        &= \sum_{\Gamma\in\mathcal{G}}\prod_{\gamma\in\Gamma}w_\gamma \notag \\
        &= Z_\mathcal{P}. \notag
    \end{align}
    We now show that $\mu_{\rho_\mathcal{S}}=\hat{\mu}_\mathcal{P}$. For any spin configuration $\sigma\in[d]^{V(G)}$, we have
    \begin{align}
        \mu_{\rho_\mathcal{S}}(\sigma) &= \frac{1}{Z_\mathcal{S}}\Tr\left[\ketbra{\sigma}{\sigma}e^{-\beta H_\mathcal{S}}\right] \notag \\
        &= \frac{1}{Z_\mathcal{P}}\sum_{\Gamma\in\mathcal{G}}\frac{1}{d^\abs{G}}\prod_{\gamma\in\Gamma}d^\abs{\gamma}(-1)^\norm{\gamma}\sum_{T \subseteq E(\gamma)}(-1)^\abs{T}\Tr\left[\ketbra{\sigma_{V(\gamma)}}{\sigma_{V(\gamma)}}e^{-\beta\sum_{e \in T}\Phi(e)}\right] \notag \\
        &= \frac{1}{Z_\mathcal{P}}\sum_{\Gamma\in\mathcal{G}}\frac{1}{d^\abs{G}}\prod_{\gamma\in\Gamma}d^\abs{\gamma}w_\gamma\mu_\gamma(\sigma_{V(\gamma)}) \notag \\
        &= \sum_{\Gamma\in\mathcal{G}}\mu_\mathcal{P}(\Gamma)\prod_{\gamma\in\Gamma}\mu_\gamma(\sigma_{V(\gamma)})\prod_{v \in V(G){\setminus}V(\Gamma)}\mu_v(\sigma_v) \notag \\
        &= \hat{\mu}_\mathcal{P}(\sigma). \notag
    \end{align}
    Finally, we show that $w_\gamma\geq0$ for all $\gamma\in\mathcal{C}$, and that $\mu_\gamma$ is a valid probability distribution on $[d]^{V(\gamma)}$ for all $\gamma\in\mathcal{C}$ with $w_\gamma>0$. Since $\mathcal{S}$ is stoquastic, all matrix elements of $\Phi(e)$ are non-positive in the spin basis for all $e \in E(G)$. Therefore, by the Taylor series and the principle of inclusion-exclusion,
    \begin{equation}
         w_\gamma = \sum_{n=\norm{\gamma}}^\infty\frac{(-\beta)^n}{n!}\sum_{\substack{\rho \in E(\gamma)^n \\ \operatorname{supp}(\rho)=E(\gamma)}}\Tr\left[\prod_{e\in\rho}\Phi(e)\right] \geq 0. \notag
    \end{equation}
    Further, for $\gamma\in\mathcal{C}$ with $w_\gamma>0$ and $\sigma\in[d]^{V(\gamma)}$,
    \begin{equation}
        \mu_\gamma(\sigma) = \frac{1}{w_\gamma}\sum_{n=\norm{\gamma}}^\infty\frac{(-\beta)^n}{n!}\sum_{\substack{\rho \in E(\gamma)^n \\ \operatorname{supp}(\rho)=E(\gamma)}}\Tr\left[\ketbra{\sigma}{\sigma}\prod_{e\in\rho}\Phi(e)\right]
        \geq 0. \notag
    \end{equation}
    By linearity of the trace, $\sum_{\sigma\in[d]^{V(\gamma)}}\mu_\gamma(\sigma)=1$, and hence $\mu_\gamma$ defines a valid probability distribution on $[d]^{V(\gamma)}$. This completes the proof.
\end{proof}

\begin{lemma}
    \label{lemma:StoquasticWeightBound}
    Fix $\Delta\in\mathbb{Z}_{\geq3}$ and $r\geq0$. Let $\mathcal{S}=(\Phi,d,\beta)$ be a stoquastic spin system on a graph $G$ of maximum degree at most $\Delta$, and let $\mathcal{P}=(\mathcal{C},w,\sim)$ be the stoquastic polymer model of $\mathcal{S}$. Let $\beta$ be such that
    \begin{equation}
        \beta \leq \frac{1}{e^{r+3}\Delta}. \notag
    \end{equation}
    Then, for all $\gamma\in\mathcal{C}$, the weight $w_\gamma$ satisfies
    \begin{equation}
        w_\gamma \leq \left(\frac{1}{2e^{r+2}(\Delta-1)}\right)^\norm{\gamma}. \notag
    \end{equation}
\end{lemma}

The proof of Lemma~\ref{lemma:StoquasticWeightBound} follows that of Ref.~\cite{mann2024algorithmic}.

\begin{proof}
    Fix a polymer $\gamma$. By the Taylor series and the principle of inclusion-exclusion,
    \begin{equation}
         w_\gamma = \sum_{n=\norm{\gamma}}^\infty\frac{(-\beta)^n}{n!}\sum_{\substack{\rho \in E(\gamma)^n \\ \operatorname{supp}(\rho)=E(\gamma)}}\Tr\left[\prod_{e\in\rho}\Phi(e)\right] \leq \sum_{n=\norm{\gamma}}^\infty\frac{\beta^n}{n!}\sum_{\substack{\rho \in E(\gamma)^n \\ \operatorname{supp}(\rho)=E(\gamma)}}\prod_{e\in\rho}\norm{\Phi(e)}. \notag
    \end{equation}
    There are precisely $\genfrac{\{}{\}}{0pt}{}{n}{\norm{\gamma}}\norm{\gamma}!$ sequences $\rho$ of length $n$ whose support is $E(\gamma)$, where $\genfrac{\{}{\}}{0pt}{}{n}{\norm{\gamma}}$ denotes the Stirling number of the second kind. Hence,
    \begin{equation}
         w_\gamma \leq \sum_{n=\norm{\gamma}}^\infty\genfrac{\{}{\}}{0pt}{}{n}{\norm{\gamma}}\frac{\norm{\gamma}!}{n!}\beta^n = (e^\beta-1)^\norm{\gamma}, \notag
    \end{equation}
    where we have used the identity $\sum_{n=k}^\infty\genfrac{\{}{\}}{0pt}{}{n}{k}\frac{x^n}{n!}=\frac{(e^x-1)^k}{k!}$. By taking $\beta\leq\frac{1}{e^{r+3}\Delta}$,
    \begin{equation}
        w_\gamma \leq \left(\frac{1}{2e^{r+2}(\Delta-1)}\right)^\norm{\gamma}, \notag
    \end{equation}
    completing the proof.
\end{proof}

\begin{lemma}
    \label{lemma:StoquasticKappaComputable}
    Let $\mathcal{S}=(\Phi,d,\beta)$ be a stoquastic spin system on a graph $G$, and let $\mathcal{P}=(\mathcal{C},w,\sim)$ be the stoquastic polymer model of $\mathcal{S}$. Then $\mathcal{P}$ is $\kappa$-computable with $\kappa=\log(2d^3)$.
\end{lemma}

\begin{proof}
    The weight $w_\gamma$ is a sum over all subsets $T$ of $E(\gamma)$, of which there are $2^\norm{\gamma}$. For each of these subsets $T$, the sum of interactions can be diagonalised in time $d^{3\norm{\gamma}}\norm{\gamma}^{O(1)}$. Hence, $w_\gamma$ can be computed and $\mu_\gamma$ can be sampled in time $e^{\log(2d^3)\norm{\gamma}}\norm{\gamma}^{O(1)}$, and so $\mathcal{P}$ is $\kappa$-computable with $\kappa=\log(2d^3)$.
\end{proof}

We now prove Theorem~\ref{theorem:StoquasticModelAlgorithm}.

\begin{proof}[Proof of Theorem~\ref*{theorem:StoquasticModelAlgorithm}]
    Let $\mathcal{P}=(\mathcal{C},w,\sim)$ be the stoquastic polymer model of $\mathcal{S}$. By Lemma~\ref{lemma:StoquasticPolymerRepresentation}, $\mu_{\rho_\mathcal{S}}=\hat{\mu}_\mathcal{P}$ and $Z_\mathcal{S}=Z_\mathcal{P}$, so it suffices to sample from $\hat{\mu}_\mathcal{P}$ and approximate $Z_\mathcal{P}$. A sample from $\hat{\mu}_\mathcal{P}$ is obtained by sampling $\Gamma\sim\mu_\mathcal{P}$ and then sampling spin configurations $\sigma_{V(\gamma)}\sim\mu_\gamma$ for each $\gamma\in\Gamma$ and $\sigma_v\sim\mu_v$ for each $v \in V(G){\setminus}V(\Gamma)$. By Lemma~\ref{lemma:StoquasticWeightBound} with $r=\log(2d^3)$, for all $\gamma\in\mathcal{C}$,
    \begin{equation}
        w_\gamma \leq \left(\frac{1}{4e^2d^3(\Delta-1)}\right)^\norm{\gamma}. \notag
    \end{equation}
    Since $\mathcal{P}$ is a subgraph polymer model on $G$ of maximum degree at most $\Delta$, it is equivalent to a graphlet polymer model on $L(G)$, which has maximum degree at most $2(\Delta-1)$. By Lemma~\ref{lemma:StoquasticKappaComputable}, $\mathcal{P}$ is $\kappa$-computable with $\kappa=\log(2d^3)$. By applying Theorem~\ref{theorem:PolymerSamplingAlgorithm} and Theorem~\ref{theorem:PolymerCountingAlgorithm} with $r=\kappa$, there is an $\epsilon$-approximate sampling algorithm for $\hat{\mu}_\mathcal{P}$ with expected runtime $O(\abs{G}\log(\abs{G}/\epsilon))$ and an $\epsilon$-approximate counting algorithm for $Z_\mathcal{P}$ with expected runtime $O(\abs{G}^2\epsilon^{-2}\log(\abs{G}/\epsilon)^2)$. This completes the proof.
\end{proof}

We now improve the inverse temperature condition of Theorem~\ref{theorem:StoquasticModelAlgorithm} by considering a probabilistic stoquastic polymer model representation. Our second result is as follows.

\begin{theorem}
    \label{theorem:ProbabilisticStoquasticModelAlgorithm}
    Fix $\Delta\in\mathbb{Z}_{\geq3}$. Let $\mathcal{S}=(\Phi,d,\beta)$ be a stoquastic spin system on a graph $G$ of maximum degree at most $\Delta$. Let $\beta$ be such that
    \begin{equation}
        \beta \leq \frac{1}{e^3d^2\Delta}. \notag
    \end{equation}
    Then, for any $\epsilon>0$, there is an $\epsilon$-approximate sampling algorithm for $\mu_{\rho_\mathcal{S}}$ with expected runtime $O(\abs{G}\log(\abs{G}/\epsilon))$ and an $\epsilon$-approximate counting algorithm for $Z_\mathcal{S}$ with expected runtime $O(\abs{G}^2\epsilon^{-2}\log(\abs{G}/\epsilon)^2)$.
\end{theorem}

We prove Theorem~\ref{theorem:ProbabilisticStoquasticModelAlgorithm} by showing that the conditions required to apply Theorem~\ref{theorem:PolymerSamplingAlgorithm} and Theorem~\ref{theorem:PolymerCountingAlgorithm} are satisfied. In particular, we show that (1) the stoquastic polymer model admits a suitable probabilistic representation, and (2) the probabilistic stoquastic polymer model is $\kappa$-computable.

Let $\mathcal{S}=(\Phi,d,\beta)$ be a stoquastic spin system on a graph $G$, and let $\mathcal{P}=(\mathcal{C},w,\sim)$ be the stoquastic polymer model of $\mathcal{S}$. For each $\gamma\in\mathcal{C}$, let $\nu_\gamma$ denote the probability distribution over sequences $\rho$ of edges of $\gamma$ with $\operatorname{supp}(\rho)=E(\gamma)$ defined by sampling $n\sim\operatorname{Poi}(\beta\norm{\gamma})$ and then $\rho$ uniformly at random from sequences of length $\norm{\gamma}+n$ with support $E(\gamma)$. We define the \emph{probabilistic stoquastic polymer model} of $\mathcal{P}$ to be the probabilistic polymer model $\mathcal{P}_\circ=(\mathcal{C},\tilde{w},\sim)$ where, for each $\gamma\in\mathcal{C}$, the weight $\tilde{w}_\gamma$ is defined for $\rho\sim\nu_\gamma$ by
\begin{equation}
    \tilde{w}_\gamma(\rho) \coloneqq \frac{(\beta e^\beta)^\norm{\gamma}(-1)^\abs{\rho}\genfrac{\{}{\}}{0pt}{}{\abs{\rho}}{\norm{\gamma}}\norm{\gamma}^\norm{\gamma}}{\binom{\abs{\rho}}{\norm{\gamma}}\norm{\gamma}^\abs{\rho}}\Tr\left[\prod_{e\in\rho}\Phi(e)\right]. \notag
\end{equation}
For each $\gamma\in\mathcal{C}$ and $\rho\in\operatorname{supp}(\nu_\gamma)$ with $\tilde{w}_\gamma(\rho)>0$, let $\tilde{\mu}_{\gamma,\rho}$ be the probability distribution on $[d]^{V(\gamma)}$ defined by
\begin{equation}
    \tilde{\mu}_{\gamma,\rho}(\sigma) \coloneqq \frac{\Tr\left[\ketbra{\sigma}{\sigma}\prod_{e\in\rho}\Phi(e)\right]}{\Tr\left[\prod_{e\in\rho}\Phi(e)\right]}. \notag
\end{equation}

\begin{lemma}
    \label{lemma:StoquasticWeightEstimator}
    Fix $\Delta\in\mathbb{Z}_{\geq3}$ and $r\geq0$. Let $\mathcal{S}=(\Phi,d,\beta)$ be a stoquastic spin system on a graph $G$ of maximum degree at most $\Delta$, let $\mathcal{P}=(\mathcal{C},w,\sim)$ be the stoquastic polymer model of $\mathcal{S}$, and let $\mathcal{P}_\circ=(\mathcal{C},\tilde{w},\sim)$ be the probabilistic stoquastic polymer model of $\mathcal{P}$. Let $\beta$ be such that
    \begin{equation}
        \beta \leq \frac{1}{e^{r+3}\Delta}. \notag
    \end{equation}
    Then, for all $\gamma\in\mathcal{C}$ and $\rho\in\operatorname{supp}(\nu_\gamma)$,
    \begin{equation}
        \mathbb{E}_{\nu_\gamma}\left[\tilde{w}_\gamma\right]=w_\gamma \quad\text{and}\quad \tilde{w}_\gamma(\rho) \leq \left(\frac{1}{2e^{r+2}(\Delta-1)}\right)^\norm{\gamma}, \notag
    \end{equation}
    and, for all $\gamma\in\mathcal{C}$ with $w_\gamma>0$ and $\sigma\in[d]^{V(\gamma)}$,
    \begin{equation}
        \mathbb{E}_{\nu_\gamma}\left[\tilde{w}_\gamma\tilde{\mu}_{\gamma,\rho}(\sigma)\right] = w_\gamma\mu_\gamma(\sigma). \notag
    \end{equation}
\end{lemma}

\begin{proof}
    Fix a polymer $\gamma\in\mathcal{C}$. We first show that $\mathbb{E}_{\nu_\gamma}\left[\tilde{w}_\gamma\right]=w_\gamma$. Since there are $\genfrac{\{}{\}}{0pt}{}{\norm{\gamma}+n}{\norm{\gamma}}\norm{\gamma}!$ sequences $\rho$ of length $\norm{\gamma}+n$ whose support is $E(\gamma)$,
    \begin{align}
        \mathbb{E}_{\nu_\gamma}[\tilde{w}_\gamma] &= \sum_{n=0}^\infty\frac{(\beta\norm{\gamma})^ne^{-\beta\norm{\gamma}}}{n!}\frac{1}{\genfrac{\{}{\}}{0pt}{}{\norm{\gamma}+n}{\norm{\gamma}}\norm{\gamma}!}\sum_{\substack{\rho \in E(\gamma)^{\norm{\gamma}+n} \\ \operatorname{supp}(\rho)=E(\gamma)}}\tilde{w}_\gamma(\rho) \notag \\
        &= \sum_{n=0}^\infty\frac{(-\beta)^{\norm{\gamma}+n}}{(\norm{\gamma}+n)!}\sum_{\substack{\rho \in E(\gamma)^{\norm{\gamma}+n} \\ \operatorname{supp}(\rho)=E(\gamma)}}\Tr\left[\prod_{e\in\rho}\Phi(e)\right] \notag \\
        &= \sum_{n=\norm{\gamma}}^\infty\frac{(-\beta)^n}{n!}\sum_{\substack{\rho \in E(\gamma)^n \\ \operatorname{supp}(\rho)=E(\gamma)}}\Tr\left[\prod_{e\in\rho}\Phi(e)\right] \notag \\
        &= w_\gamma, \notag
    \end{align}
    where the last equality follows by the Taylor series and the principle of inclusion-exclusion. We now show that $\tilde{w}_\gamma(\rho)\leq\left(\frac{1}{2e^{r+2}(\Delta-1)}\right)^\norm{\gamma}$ for all $\rho\in\operatorname{supp}(\nu_\gamma)$. Since $\mathcal{S}$ is stoquastic and $\norm{\Phi(e)}\leq1$ for all $e\in E(G)$,
    \begin{align}
        \tilde{w}_\gamma(\rho) &= \frac{(\beta e^\beta)^\norm{\gamma}(-1)^\abs{\rho}\genfrac{\{}{\}}{0pt}{}{\abs{\rho}}{\norm{\gamma}}\norm{\gamma}^\norm{\gamma}}{\binom{\abs{\rho}}{\norm{\gamma}}\norm{\gamma}^\abs{\rho}}\Tr\left[\prod_{e\in\rho}\Phi(e)\right] \notag \\
        &\leq \frac{(\beta e^\beta)^\norm{\gamma}\genfrac{\{}{\}}{0pt}{}{\abs{\rho}}{\norm{\gamma}}\norm{\gamma}^\norm{\gamma}}{\binom{\abs{\rho}}{\norm{\gamma}}\norm{\gamma}^\abs{\rho}} \notag \\
        &\leq (\beta e^\beta)^\norm{\gamma}, \notag
    \end{align}
    where we have used $\genfrac{\{}{\}}{0pt}{}{\abs{\rho}}{\norm{\gamma}}\norm{\gamma}^\norm{\gamma}\leq\binom{\abs{\rho}}{\norm{\gamma}}\norm{\gamma}^\abs{\rho}$. By taking $\beta\leq\frac{1}{e^{r+3}\Delta}$,
    \begin{equation}
        \tilde{w}_\gamma(\rho) \leq \left(\frac{1}{2e^{r+2}(\Delta-1)}\right)^\norm{\gamma}. \notag
    \end{equation}
    Finally, we show that $\mathbb{E}_{\nu_\gamma}\left[\tilde{w}_\gamma\tilde{\mu}_{\gamma,\rho}(\sigma)\right]=w_\gamma\mu_\gamma(\sigma)$ for all $\gamma\in\mathcal{C}$ with $w_\gamma>0$ and $\sigma\in[d]^{V(\gamma)}$. By a similar argument as for $\mathbb{E}_{\nu_\gamma}[\tilde{w}_\gamma]=w_\gamma$,
    \begin{align}
        \mathbb{E}_{\nu_\gamma}\left[\tilde{w}_\gamma\tilde{\mu}_{\gamma,\rho}(\sigma)\right] &= \sum_{n=\norm{\gamma}}^\infty\frac{(-\beta)^n}{n!}\sum_{\substack{\rho \in E(\gamma)^n \\ \operatorname{supp}(\rho)=E(\gamma)}}\Tr\left[\ketbra{\sigma}{\sigma}\prod_{e\in\rho}\Phi(e)\right] \notag \\
        &= w_\gamma\mu_\gamma(\sigma). \notag
    \end{align}
    This completes the proof.
\end{proof}

\begin{lemma}
    \label{lemma:ProbabilisticStoquasticKappaComputable}
    Let $\mathcal{S}=(\Phi,d,\beta)$ be a stoquastic spin system on a graph $G$, let $\mathcal{P}=(\mathcal{C},w,\sim)$ be the stoquastic polymer model of $\mathcal{S}$, and let $\mathcal{P}_\circ=(\mathcal{C},\tilde{w},\sim)$ be the probabilistic stoquastic polymer model of $\mathcal{P}$. Then $\mathcal{P}_\circ$ is $\kappa$-computable with $\kappa=2\log(d)$.
\end{lemma}

\begin{proof}
    A sample $\rho\sim\nu_\gamma$ can be obtained in expected time $\norm{\gamma}^{O(1)}$ by sampling $n\sim\operatorname{Poi}(\beta\norm{\gamma})$ and then sampling a uniform surjection of length $\norm{\gamma}+n$ onto $E(\gamma)$. For such a sample $\rho$, the product of interactions can be applied to each of the $d^\abs{\gamma}$ basis vectors of dimension $d^\abs{\gamma}$ in expected time $d^{2\norm{\gamma}}\norm{\gamma}^{O(1)}$. Hence, $\tilde{w}_\gamma(\rho)$ can be computed and $\tilde{\mu}_{\gamma,\rho}$ can be sampled in expected time $e^{2\log(d)\norm{\gamma}}\norm{\gamma}^{O(1)}$, and so $\mathcal{P}_\circ$ is $\kappa$-computable with $\kappa=2\log(d)$.
\end{proof}

We now prove Theorem~\ref{theorem:ProbabilisticStoquasticModelAlgorithm}.

\begin{proof}[Proof of Theorem~\ref*{theorem:ProbabilisticStoquasticModelAlgorithm}]
    Let $\mathcal{P}=(\mathcal{C},w,\sim)$ be the stoquastic polymer model of $\mathcal{S}$ and let $\mathcal{P}_\circ=(\mathcal{C},\tilde{w},\sim)$ be the probabilistic stoquastic polymer model of $\mathcal{P}$. By Lemma~\ref{lemma:StoquasticPolymerRepresentation}, $\mu_{\rho_\mathcal{S}}=\hat{\mu}_\mathcal{P}$ and $Z_\mathcal{S}=Z_\mathcal{P}$, so it suffices to sample from $\hat{\mu}_\mathcal{P}$ and approximate $Z_\mathcal{P}$. A sample from $\hat{\mu}_\mathcal{P}$ is obtained by sampling $\Gamma\sim\mu_\mathcal{P}$ and then sampling spin configurations $\sigma_{V(\gamma)}\sim\tilde{\mu}_{\gamma,\rho_\gamma}$ for each $\gamma\in\Gamma$ and $\sigma_v\sim\mu_v$ for each $v \in V(G){\setminus}V(\Gamma)$, where $\rho_\gamma$ denotes the sequence sampled by the accepting iteration of the graphlet sampler that produced $\gamma$. By Lemma~\ref{lemma:StoquasticWeightEstimator} with $r=2\log(d)$, for all $\gamma\in\mathcal{C}$ and $\rho\in\operatorname{supp}(\nu_\gamma)$,
    \begin{equation}
        \mathbb{E}_{\nu_\gamma}\left[\tilde{w}_\gamma\right]=w_\gamma \quad\text{and}\quad \tilde{w}_\gamma(\rho) \leq \left(\frac{1}{2e^2d^2(\Delta-1)}\right)^\norm{\gamma}, \notag
    \end{equation}
    and, for all $\gamma\in\mathcal{C}$ with $w_\gamma>0$ and $\sigma\in[d]^{V(\gamma)}$,
    \begin{equation}
        \mathbb{E}_{\nu_\gamma}\left[\tilde{w}_\gamma\tilde{\mu}_{\gamma,\rho}(\sigma)\right] = w_\gamma\mu_\gamma(\sigma). \notag
    \end{equation}
    In particular, $\mathcal{P}_\circ$ satisfies the conditions of Lemma~\ref{lemma:ProbabilisticGraphletSampler}. Since $\mathcal{P}$ is a subgraph polymer model on $G$ of maximum degree at most $\Delta$, it is equivalent to a graphlet polymer model on $L(G)$, which has maximum degree at most $2(\Delta-1)$. By Lemma~\ref{lemma:ProbabilisticStoquasticKappaComputable}, $\mathcal{P}_\circ$ is $\kappa$-computable with $\kappa=2\log(d)$. By applying Theorem~\ref{theorem:PolymerSamplingAlgorithm} and Theorem~\ref{theorem:PolymerCountingAlgorithm} with $r=\kappa$, there is an $\epsilon$-approximate sampling algorithm for $\hat{\mu}_\mathcal{P}$ with expected runtime $O(\abs{G}\log(\abs{G}/\epsilon))$ and an $\epsilon$-approximate counting algorithm for $Z_\mathcal{P}$ with expected runtime $O(\abs{G}^2\epsilon^{-2}\log(\abs{G}/\epsilon)^2)$. This completes the proof.
\end{proof}

We now improve the inverse temperature condition of Theorem~\ref{theorem:ProbabilisticStoquasticModelAlgorithm} by considering a truncated polymer model approach, at the cost of increased runtime.

\begin{corollary}
    \label{corollary:TruncatedStoquasticModelAlgorithm}
    Fix $\Delta\in\mathbb{Z}_{\geq3}$ and $r\geq0$. Let $\mathcal{S}=(\Phi,d,\beta)$ be a stoquastic spin system on a graph $G$ of maximum degree at most $\Delta$. Let $\beta$ be such that
    \begin{equation}
        \beta \leq \frac{1}{e^{r+3}\Delta}. \notag
    \end{equation}
    Then, for any $\epsilon>0$, there is an $\epsilon$-approximate sampling algorithm for $\mu_{\rho_\mathcal{S}}$ and an $\epsilon$-approximate counting algorithm for $Z_\mathcal{S}$. The sampling algorithm has expected runtime $O(\abs{G}\log(\abs{G}/\epsilon))$ for $r\geq2\log(d)$ and expected runtime $\abs{G}^{2\log(d)-r+1}\epsilon^{-(2\log(d)-r)}\log(\abs{G}/\epsilon)^{O(1)}$ for $r<2\log(d)$. The counting algorithm has expected runtime $O(\abs{G}^2\epsilon^{-2}\log(\abs{G}/\epsilon)^2)$ for $r\geq2\log(d)$ and expected runtime $\abs{G}^{2\log(d)-r+2}\epsilon^{-(2\log(d)-r+2)}\log(\abs{G}/\epsilon)^{O(1)}$ for $r<2\log(d)$.
\end{corollary}

\begin{proof}
    The proof follows that of Theorem~\ref{theorem:ProbabilisticStoquasticModelAlgorithm} with Corollary~\ref{corollary:TruncatedPolymerAlgorithm} replacing Theorem~\ref{theorem:PolymerSamplingAlgorithm} and Theorem~\ref{theorem:PolymerCountingAlgorithm}.
\end{proof}

\subsection{Ferromagnetic Heisenberg Models}
\label{section:FerromagneticHeisenbergModels}

We now consider ferromagnetic Heisenberg models. Let $\mathcal{S}=(\beta,h)$ denote a ferromagnetic Heisenberg model on a graph $G$ with inverse temperature $\beta\geq0$ and field strength $h\in\mathbb{R}$, defined by the Hamiltonian
\begin{equation}
    H_\mathcal{S} \coloneqq -\sum_{e \in E(G)}T_e-h\sum_{v \in V(G)}Z_v, \notag
\end{equation}
where $T_e$ denotes the operator that transposes the spins at the endpoints of $e$, and $Z_v$ denotes the Pauli $Z$ operator on the spin at vertex $v$. Our main result for the ferromagnetic Heisenberg model is as follows.

\begin{theorem}
    \label{theorem:FerromagneticHeisenbergModelAlgorithm}
    Fix $\Delta\in\mathbb{Z}_{\geq3}$. Let $\mathcal{S}=(\beta,h)$ be a ferromagnetic Heisenberg model on a graph $G$ of maximum degree at most $\Delta$. Let $\beta$ be such that
    \begin{equation}
        \beta \leq \frac{1}{e^3\Delta}. \notag
    \end{equation}
    Then, for any $\epsilon>0$, there is an $\epsilon$-approximate sampling algorithm for $\mu_{\rho_\mathcal{S}}$ with expected runtime $O(\abs{G}\log(\abs{G}/\epsilon))$ and an $\epsilon$-approximate counting algorithm for $Z_\mathcal{S}$ with expected runtime $O(\abs{G}^2\epsilon^{-2}\log(\abs{G}/\epsilon)^2)$.
\end{theorem}

We prove Theorem~\ref{theorem:FerromagneticHeisenbergModelAlgorithm} by showing that the conditions required to apply Theorem~\ref{theorem:PolymerSamplingAlgorithm} and Theorem~\ref{theorem:PolymerCountingAlgorithm} are satisfied. In particular, we show that (1) the ferromagnetic Heisenberg model admits a suitable subgraph polymer model representation, (2) the Heisenberg polymer model admits a suitable probabilistic representation, and (3) the probabilistic Heisenberg polymer model is $\kappa$-computable.

We now introduce a polymer model representation of the Heisenberg model based on the cycle representation of T\'oth~\cite{toth1993improved, goldschmidt2011quantum}. Let $\mathcal{S}=(\beta,h)$ be a ferromagnetic Heisenberg model on a graph $G$. We define the \emph{Heisenberg polymer model} of $\mathcal{S}$ to be the subgraph polymer model $\mathcal{P}=(\mathcal{C},w,\sim)$ on $G$, where $\mathcal{C}$ is the set of all connected subgraphs of $G$ with at least one edge, and $w:\mathcal{C}\to\mathbb{R}_{\geq0}$ is defined by
\begin{equation}
    w_\gamma \coloneqq \sum_{\rho \in P_\gamma}\frac{\beta^\abs{\rho}}{\abs{\rho}!}\prod_{c \in C_\rho}\frac{2\cosh(\beta h\abs{c})}{\left(2\cosh(\beta h)\right)^\abs{c}}, \notag
\end{equation}
where $P_\gamma$ denotes the set of all sequences $\rho$ of edges of $\gamma$ with $\operatorname{supp}(\rho)=E(\gamma)$, and $C_\rho$ denotes the cycles of the permutation on $V(\gamma)$ induced by $\prod_{e\in\rho}T_e$. For each $\gamma\in\mathcal{C}$ with $w_\gamma>0$, let $\mu_\gamma$ be the probability distribution on $\{-1,+1\}^{V(\gamma)}$ defined by
\begin{equation}
    \mu_\gamma(\sigma) \coloneqq \frac{1}{w_\gamma}\sum_{\rho \in P_\gamma}\frac{\beta^\abs{\rho}}{\abs{\rho}!}\prod_{c \in C_\rho}\frac{e^{\beta h\sum_{v \in c}\sigma_v}\mathds{1}_\text{$\sigma_c$ constant}}{\left(2\cosh(\beta h)\right)^\abs{c}}. \notag
\end{equation}
For each $v \in V(G)$, let $\mu_v$ be the probability distribution on $\{-1,+1\}$ defined by $\mu_v(\sigma)=\frac{e^{\beta h\sigma}}{2\cosh(\beta h)}$ for all $\sigma\in\{-1,+1\}$.

\begin{lemma}
    \label{lemma:FerromagneticHeisenbergPolymerRepresentation}
    Let $\mathcal{S}=(\beta,h)$ be a ferromagnetic Heisenberg model on a graph $G$, and let $\mathcal{P}=(\mathcal{C},w,\sim)$ be the Heisenberg polymer model of $\mathcal{S}$. Then the partition function of $\mathcal{S}$ satisfies $Z_\mathcal{S}=\cosh(\beta h)^\abs{G}Z_\mathcal{P}$ and the thermal distribution of $\mathcal{S}$ satisfies $\mu_{\rho_\mathcal{S}}=\hat{\mu}_\mathcal{P}$.
\end{lemma}

\begin{proof}
    We first show that $Z_\mathcal{S}=\cosh(\beta h)^\abs{G}Z_\mathcal{P}$. By the principle of inclusion-exclusion and commutativity,
    \begin{align}
        Z_\mathcal{S} &= \Tr\left[e^{-\beta H_\mathcal{S}}\right] \notag \\
        &= \sum_{S \subseteq E(G)}(-1)^\abs{S}\sum_{T \subseteq S}(-1)^\abs{T}\Tr\left[e^{\beta h\sum_{v \in V(G)}Z_v}e^{\beta\sum_{e \in T}T_e}\right]. \notag
    \end{align}
    For a subset $S \subseteq E(G)$, let $\Gamma_S$ denote the maximally connected components of $S$. By factorising over these components,
    \begin{align}
        Z_\mathcal{S} &= \Tr\left[e^{\beta h\sum_{v \in V(G)}Z_v}\right]\sum_{S \subseteq E(G)}\prod_{\gamma\in\Gamma_S}(-1)^\norm{\gamma}\sum_{T \subseteq E(\gamma)}(-1)^\abs{T}\frac{\Tr\left[e^{\beta h\sum_{v \in V(\gamma)}Z_v}e^{\beta\sum_{e \in T}T_e}\right]}{\Tr\left[e^{\beta h\sum_{v \in V(\gamma)}Z_v}\right]} \notag \\
        &= \cosh(\beta h)^\abs{G}\sum_{S \subseteq E(G)}\prod_{\gamma\in\Gamma_S}(-1)^\norm{\gamma}\sum_{T \subseteq E(\gamma)}(-1)^\abs{T}\frac{\Tr\left[e^{\beta h\sum_{v \in V(\gamma)}Z_v}e^{\beta\sum_{e \in T}T_e}\right]}{\cosh(\beta h)^\abs{\gamma}}. \notag
    \end{align}
    By the Taylor series and the principle of inclusion-exclusion,
    \begin{align}
        Z_\mathcal{S} &= \cosh(\beta h)^\abs{G}\sum_{S \subseteq E(G)}\prod_{\gamma\in\Gamma_S}\sum_{n=\norm{\gamma}}^\infty\frac{\beta^n}{n!}\sum_{\substack{\rho \in E(\gamma)^n \\ \operatorname{supp}(\rho)=E(\gamma)}}\frac{\Tr\left[e^{\beta h\sum_{v \in V(\gamma)}Z_v}\prod_{e\in\rho}T_e\right]}{\cosh(\beta h)^\abs{\gamma}} \notag \\
        &= \cosh(\beta h)^\abs{G}\sum_{S \subseteq E(G)}\prod_{\gamma\in\Gamma_S}\sum_{\rho \in P_\gamma}\frac{\beta^\abs{\rho}\Tr\left[e^{\beta h\sum_{v \in V(\gamma)}Z_v}\prod_{e\in\rho}T_e\right]}{\abs{\rho}!\cosh(\beta h)^\abs{\gamma}}. \notag
    \end{align}
    Since the trace is non-zero only when $\sigma$ is constant on each cycle $c \in C_\rho$,
    \begin{align}
         Z_\mathcal{S} &= \cosh(\beta h)^\abs{G}\sum_{S \subseteq E(G)}\prod_{\gamma\in\Gamma_S}\sum_{\rho \in P_\gamma}\frac{\beta^\abs{\rho}}{\abs{\rho}!}\prod_{c \in C_\rho}\frac{2\cosh(\beta h\abs{c})}{\left(2\cosh(\beta h)\right)^\abs{c}} \notag \\
         &= \cosh(\beta h)^\abs{G}\sum_{S \subseteq E(G)}\prod_{\gamma\in\Gamma_S}w_\gamma \notag \\
         &= \cosh(\beta h)^\abs{G}\sum_{\Gamma\in\mathcal{G}}\prod_{\gamma\in\Gamma}w_\gamma \notag \\
         &= \cosh(\beta h)^\abs{G}Z_\mathcal{P}. \notag
    \end{align}
    We now show that $\mu_{\rho_\mathcal{S}}=\hat{\mu}_\mathcal{P}$. For any spin configuration $\sigma\in\{-1,+1\}^{V(G)}$, we have
    \begin{align}
        \mu_{\rho_\mathcal{S}}(\sigma) &= \frac{1}{Z_\mathcal{S}}\Tr\left[\ketbra{\sigma}{\sigma}e^{-\beta H_\mathcal{S}}\right] \notag \\
        &= \frac{1}{Z_\mathcal{P}}\sum_{\Gamma\in\mathcal{G}}\frac{e^{\beta h\sum_{v \in V(G)}\sigma_v}}{\left(2\cosh(\beta h)\right)^\abs{G}}\prod_{\gamma\in\Gamma}\sum_{\rho \in P_\gamma}\frac{2^\abs{\gamma}\beta^\abs{\rho}}{\abs{\rho}!}\Tr\left[\ketbra{\sigma_{V(\gamma)}}{\sigma_{V(\gamma)}}\prod_{e\in\rho}T_e\right] \notag \\
        &= \frac{1}{Z_\mathcal{P}}\sum_{\Gamma\in\mathcal{G}}\frac{e^{\beta h\sum_{v \in V(G)}\sigma_v}}{\left(2\cosh(\beta h)\right)^\abs{G}}\prod_{\gamma\in\Gamma}\sum_{\rho \in P_\gamma}\frac{\beta^\abs{\rho}}{\abs{\rho}!}\prod_{c \in C_\rho}\mathds{1}_\text{$\sigma_c$ constant} \notag \\
        &= \frac{1}{Z_\mathcal{P}}\sum_{\Gamma\in\mathcal{G}}\frac{e^{\beta h\sum_{v \in V(G)}\sigma_v}}{\left(2\cosh(\beta h)\right)^\abs{G}}\prod_{\gamma\in\Gamma}\frac{\left(2\cosh(\beta h)\right)^\abs{\gamma}}{e^{\beta h\sum_{v \in V(\gamma)}\sigma_v}}w_\gamma\mu_\gamma(\sigma_{V(\gamma)}) \notag \\
        &= \sum_{\Gamma\in\mathcal{G}}\mu_\mathcal{P}(\Gamma)\prod_{\gamma\in\Gamma}\mu_\gamma(\sigma_{V(\gamma)})\prod_{v \in V(G){\setminus}V(\Gamma)}\mu_v(\sigma_v) \notag \\
        &= \hat{\mu}_\mathcal{P}(\sigma). \notag
    \end{align}
    Finally, we show that $w_\gamma\geq0$ for all $\gamma\in\mathcal{C}$, and that $\mu_\gamma$ is a valid probability distribution on $\{-1,+1\}^{V(\gamma)}$ for all $\gamma\in\mathcal{C}$ with $w_\gamma>0$. Since all terms in $w_\gamma$ and $\mu_\gamma$ are non-negative, we have $w_\gamma\geq0$ for all $\gamma\in\mathcal{C}$, and $\mu_\gamma(\sigma)\geq0$ for all $\gamma\in\mathcal{C}$ with $w_\gamma>0$ and $\sigma\in\{-1,+1\}^{V(\gamma)}$. By using $\sum_{\sigma\in\{-1,+1\}^c}e^{\beta h\sum_{v \in c}\sigma_v}\mathds{1}_\text{$\sigma$ constant}=2\cosh(\beta h\abs{c})$ for each $c \in C_\rho$, we have $\sum_{\sigma\in\{-1,+1\}^{V(\gamma)}}\mu_\gamma(\sigma)=1$, and hence $\mu_\gamma$ defines a valid probability distribution on $\{-1,+1\}^{V(\gamma)}$. This completes the proof.
\end{proof}

Let $\mathcal{S}=(\beta,h)$ be a ferromagnetic Heisenberg model on a graph $G$, and let $\mathcal{P}=(\mathcal{C},w,\sim)$ be the Heisenberg polymer model of $\mathcal{S}$. For each $\gamma\in\mathcal{C}$, let $\nu_\gamma$ denote the probability distribution over sequences $\rho$ of edges of $\gamma$ with $\operatorname{supp}(\rho)=E(\gamma)$ defined by sampling $n\sim\operatorname{Poi}(\beta\norm{\gamma})$ and then $\rho$ uniformly at random from sequences of length $\norm{\gamma}+n$ with support $E(\gamma)$. We define the \emph{probabilistic Heisenberg polymer model} of $\mathcal{P}$ to be the probabilistic polymer model $\mathcal{P}_\circ=(\mathcal{C},\tilde{w},\sim)$ where, for each $\gamma\in\mathcal{C}$, the weight $\tilde{w}_\gamma$ is defined for $\rho\sim\nu_\gamma$ by
\begin{equation}
    \tilde{w}_\gamma(\rho) \coloneqq \frac{(\beta e^\beta)^\norm{\gamma}\genfrac{\{}{\}}{0pt}{}{\abs{\rho}}{\norm{\gamma}}\norm{\gamma}^\norm{\gamma}}{\binom{\abs{\rho}}{\norm{\gamma}}\norm{\gamma}^\abs{\rho}}\prod_{c \in C_\rho}\frac{2\cosh(\beta h\abs{c})}{\left(2\cosh(\beta h)\right)^\abs{c}}. \notag
\end{equation}
For each $\gamma\in\mathcal{C}$ and $\rho\in\operatorname{supp}(\nu_\gamma)$ with $\tilde{w}_\gamma(\rho)>0$, let $\tilde{\mu}_{\gamma,\rho}$ be the probability distribution on $\{-1,+1\}^{V(\gamma)}$ defined by
\begin{equation}
    \tilde{\mu}_{\gamma,\rho}(\sigma) \coloneqq \prod_{c \in C_\rho}\frac{e^{\beta h\sum_{v \in c}\sigma_v}\mathds{1}_\text{$\sigma_c$ constant}}{2\cosh(\beta h\abs{c})}. \notag
\end{equation}

\begin{lemma}
    \label{lemma:FerromagneticHeisenbergWeightEstimator}
    Fix $\Delta\in\mathbb{Z}_{\geq3}$ and $r\geq0$. Let $\mathcal{S}=(\beta,h)$ be a ferromagnetic Heisenberg model on a graph $G$ of maximum degree at most $\Delta$, let $\mathcal{P}=(\mathcal{C},w,\sim)$ be the Heisenberg polymer model of $\mathcal{S}$, and let $\mathcal{P}_\circ=(\mathcal{C},\tilde{w},\sim)$ be the probabilistic Heisenberg polymer model of $\mathcal{P}$. Let $\beta$ be such that
    \begin{equation}
        \beta \leq \frac{1}{e^{r+3}\Delta}. \notag
    \end{equation}
    Then, for all $\gamma\in\mathcal{C}$ and $\rho\in\operatorname{supp}(\nu_\gamma)$,
    \begin{equation}
        \mathbb{E}_{\nu_\gamma}\left[\tilde{w}_\gamma\right]=w_\gamma \quad\text{and}\quad \tilde{w}_\gamma(\rho) \leq \left(\frac{1}{2e^{r+2}(\Delta-1)}\right)^\norm{\gamma}, \notag
    \end{equation}
    and, for all $\gamma\in\mathcal{C}$ with $w_\gamma>0$ and $\sigma\in\{-1,+1\}^{V(\gamma)}$,
    \begin{equation}
        \mathbb{E}_{\nu_\gamma}\left[\tilde{w}_\gamma\tilde{\mu}_{\gamma,\rho}(\sigma)\right] = w_\gamma\mu_\gamma(\sigma). \notag
    \end{equation}
\end{lemma}

\begin{proof}
    Fix a polymer $\gamma\in\mathcal{C}$. We first show that $\mathbb{E}_{\nu_\gamma}\left[\tilde{w}_\gamma\right]=w_\gamma$. Since there are $\genfrac{\{}{\}}{0pt}{}{\norm{\gamma}+n}{\norm{\gamma}}\norm{\gamma}!$ sequences $\rho$ of length $\norm{\gamma}+n$ whose support is $E(\gamma)$,
    \begin{align}
        \mathbb{E}_{\nu_\gamma}[\tilde{w}_\gamma] &= \sum_{n=0}^\infty\frac{(\beta\norm{\gamma})^ne^{-\beta\norm{\gamma}}}{n!}\frac{1}{\genfrac{\{}{\}}{0pt}{}{\norm{\gamma}+n}{\norm{\gamma}}\norm{\gamma}!}\sum_{\substack{\rho \in E(\gamma)^{\norm{\gamma}+n} \\ \operatorname{supp}(\rho)=E(\gamma)}}\tilde{w}_\gamma(\rho) \notag \\
        &= \sum_{n=0}^\infty\frac{\beta^{\norm{\gamma}+n}}{(\norm{\gamma}+n)!}\sum_{\substack{\rho \in E(\gamma)^{\norm{\gamma}+n} \\ \operatorname{supp}(\rho)=E(\gamma)}}\prod_{c \in C_\rho}\frac{2\cosh(\beta h\abs{c})}{\left(2\cosh(\beta h)\right)^\abs{c}} \notag \\
        &= \sum_{\rho \in P_\gamma}\frac{\beta^\abs{\rho}}{\abs{\rho}!}\prod_{c \in C_\rho}\frac{2\cosh(\beta h\abs{c})}{\left(2\cosh(\beta h)\right)^\abs{c}} \notag \\
        &= w_\gamma. \notag
    \end{align}
    We now show that $\tilde{w}_\gamma(\rho)\leq\left(\frac{1}{2e^{r+2}(\Delta-1)}\right)^\norm{\gamma}$ for all $\rho\in\operatorname{supp}(\nu_\gamma)$. Since $2\cosh(nx)\leq\left(2\cosh(x)\right)^n$ for all $x\in\mathbb{R}$ and $n\in\mathbb{Z}^+$,
    \begin{align}
        \tilde{w}_\gamma(\rho) &= \frac{(\beta e^\beta)^\norm{\gamma}\genfrac{\{}{\}}{0pt}{}{\abs{\rho}}{\norm{\gamma}}\norm{\gamma}^\norm{\gamma}}{\binom{\abs{\rho}}{\norm{\gamma}}\norm{\gamma}^\abs{\rho}}\prod_{c \in C_\rho}\frac{2\cosh(\beta h\abs{c})}{\left(2\cosh(\beta h)\right)^\abs{c}} \notag \\
        &\leq \frac{(\beta e^\beta)^\norm{\gamma}\genfrac{\{}{\}}{0pt}{}{\abs{\rho}}{\norm{\gamma}}\norm{\gamma}^\norm{\gamma}}{\binom{\abs{\rho}}{\norm{\gamma}}\norm{\gamma}^\abs{\rho}} \notag \\
        &\leq (\beta e^\beta)^\norm{\gamma}, \notag
    \end{align}
    where we have used $\genfrac{\{}{\}}{0pt}{}{\abs{\rho}}{\norm{\gamma}}\norm{\gamma}^\norm{\gamma}\leq\binom{\abs{\rho}}{\norm{\gamma}}\norm{\gamma}^\abs{\rho}$. By taking $\beta\leq\frac{1}{e^{r+3}\Delta}$,
    \begin{equation}
        \tilde{w}_\gamma(\rho) \leq \left(\frac{1}{2e^{r+2}(\Delta-1)}\right)^\norm{\gamma}. \notag
    \end{equation}
    Finally, we show that $\mathbb{E}_{\nu_\gamma}\left[\tilde{w}_\gamma\tilde{\mu}_{\gamma,\rho}(\sigma)\right]=w_\gamma\mu_\gamma(\sigma)$ for all $\gamma\in\mathcal{C}$ with $w_\gamma>0$ and $\sigma\in\{-1,+1\}^{V(\gamma)}$. By a similar argument as for $\mathbb{E}_{\nu_\gamma}[\tilde{w}_\gamma]=w_\gamma$,
    \begin{align}
        \mathbb{E}_{\nu_\gamma}\left[\tilde{w}_\gamma\tilde{\mu}_{\gamma,\rho}(\sigma)\right] &= \sum_{\rho \in P_\gamma}\frac{\beta^\abs{\rho}}{\abs{\rho}!}\prod_{c \in C_\rho}\frac{e^{\beta h\sum_{v \in c}\sigma_v}\mathds{1}_\text{$\sigma_c$ constant}}{\left(2\cosh(\beta h)\right)^\abs{c}} \notag \\
        &= w_\gamma\mu_\gamma(\sigma). \notag
    \end{align}
    This completes the proof.
\end{proof}

\begin{lemma}
    \label{lemma:ProbabilisticFerromagneticHeisenbergKappaComputable}
    Let $\mathcal{S}=(\beta,h)$ be a ferromagnetic Heisenberg model on a graph $G$, let $\mathcal{P}=(\mathcal{C},w,\sim)$ be the Heisenberg polymer model of $\mathcal{S}$, and let $\mathcal{P}_\circ=(\mathcal{C},\tilde{w},\sim)$ be the probabilistic Heisenberg polymer model of $\mathcal{P}$. Then $\mathcal{P}_\circ$ is $\kappa$-computable with $\kappa=0$.
\end{lemma}

\begin{proof}
    A sample $\rho\sim\nu_\gamma$ can be obtained in expected time $\norm{\gamma}^{O(1)}$ by sampling $n\sim\operatorname{Poi}(\beta\norm{\gamma})$ and then sampling a uniform surjection of length $\norm{\gamma}+n$ onto $E(\gamma)$. For such a sample $\rho$, the permutation on $V(\gamma)$ induced by $\prod_{e\in\rho}T_e$ can be computed in expected time $\norm{\gamma}^{O(1)}$, and the cycle decomposition $C_\rho$ can then be obtained in time $\abs{\gamma}^{O(1)}$. Hence, $\tilde{w}_\gamma(\rho)$ can be computed and $\tilde{\mu}_{\gamma,\rho}$ can be sampled in expected time $\norm{\gamma}^{O(1)}$, and so $\mathcal{P}_\circ$ is $\kappa$-computable with $\kappa=0$.
\end{proof}

We now prove Theorem~\ref{theorem:FerromagneticHeisenbergModelAlgorithm}.

\begin{proof}[Proof of Theorem~\ref*{theorem:FerromagneticHeisenbergModelAlgorithm}]
    Let $\mathcal{P}=(\mathcal{C},w,\sim)$ be the Heisenberg polymer model of $\mathcal{S}$ and let $\mathcal{P}_\circ=(\mathcal{C},\tilde{w},\sim)$ be the probabilistic Heisenberg polymer model of $\mathcal{P}$. By Lemma~\ref{lemma:FerromagneticHeisenbergPolymerRepresentation}, $\mu_{\rho_\mathcal{S}}=\hat{\mu}_\mathcal{P}$ and $Z_\mathcal{S}=\cosh(\beta h)^\abs{G}Z_\mathcal{P}$, so it suffices to sample from $\hat{\mu}_\mathcal{P}$ and approximate $Z_\mathcal{P}$. A sample from $\hat{\mu}_\mathcal{P}$ is obtained by sampling $\Gamma\sim\mu_\mathcal{P}$ and then sampling spin configurations $\sigma_{V(\gamma)}\sim\tilde{\mu}_{\gamma,\rho_\gamma}$ for each $\gamma\in\Gamma$ and $\sigma_v\sim\mu_v$ for each $v \in V(G){\setminus}V(\Gamma)$, where $\rho_\gamma$ denotes the sequence sampled by the accepting iteration of the graphlet sampler that produced $\gamma$. By Lemma~\ref{lemma:FerromagneticHeisenbergWeightEstimator} with $r=0$, for all $\gamma\in\mathcal{C}$ and $\rho\in\operatorname{supp}(\nu_\gamma)$,
    \begin{equation}
        \mathbb{E}_{\nu_\gamma}\left[\tilde{w}_\gamma\right]=w_\gamma \quad\text{and}\quad \tilde{w}_\gamma(\rho) \leq \left(\frac{1}{2e^2(\Delta-1)}\right)^\norm{\gamma}, \notag
    \end{equation}
    and, for all $\gamma\in\mathcal{C}$ with $w_\gamma>0$ and $\sigma\in\{-1,+1\}^{V(\gamma)}$,
    \begin{equation}
        \mathbb{E}_{\nu_\gamma}\left[\tilde{w}_\gamma\tilde{\mu}_{\gamma,\rho}(\sigma)\right] = w_\gamma\mu_\gamma(\sigma). \notag
    \end{equation}
    In particular, $\mathcal{P}_\circ$ satisfies the conditions of Lemma~\ref{lemma:ProbabilisticGraphletSampler}. Since $\mathcal{P}$ is a subgraph polymer model on $G$ of maximum degree at most $\Delta$, it is equivalent to a graphlet polymer model on $L(G)$, which has maximum degree at most $2(\Delta-1)$. By Lemma~\ref{lemma:ProbabilisticFerromagneticHeisenbergKappaComputable}, $\mathcal{P}_\circ$ is $\kappa$-computable with $\kappa=0$. By applying Theorem~\ref{theorem:PolymerSamplingAlgorithm} and Theorem~\ref{theorem:PolymerCountingAlgorithm} with $r=\kappa$, there is an $\epsilon$-approximate sampling algorithm for $\hat{\mu}_\mathcal{P}$ with expected runtime $O(\abs{G}\log(\abs{G}/\epsilon))$ and an $\epsilon$-approximate counting algorithm for $Z_\mathcal{P}$ with expected runtime $O(\abs{G}^2\epsilon^{-2}\log(\abs{G}/\epsilon)^2)$. This completes the proof.
\end{proof}

\subsection{Antiferromagnetic Heisenberg Models}
\label{section:AntiferromagneticHeisenbergModels}

We finally consider antiferromagnetic Heisenberg models on bipartite graphs. Let $\mathcal{S}=(\beta,h)$ denote an antiferromagnetic Heisenberg model on a graph $G$ with inverse temperature $\beta\geq0$ and field strength $h\in\mathbb{R}$, defined by the Hamiltonian
\begin{equation}
    H_\mathcal{S} \coloneqq \sum_{e \in E(G)}T_e-h\sum_{v \in V(G)}Z_v, \notag
\end{equation}
where $T_e$ denotes the operator that transposes the spins at the endpoints of $e$, and $Z_v$ denotes the Pauli $Z$ operator on the spin at vertex $v$. Our main result for the antiferromagnetic Heisenberg model on bipartite graphs is as follows.

\begin{theorem}
    \label{theorem:AntiferromagneticBipartiteHeisenbergModelAlgorithm}
    Fix $\Delta\in\mathbb{Z}_{\geq3}$. Let $\mathcal{S}=(\beta,h)$ be an antiferromagnetic Heisenberg model on a bipartite graph $G$ of maximum degree at most $\Delta$. Let $\beta$ be such that
    \begin{equation}
        \beta \leq \frac{1}{2e^3\Delta}. \notag
    \end{equation}
    Then, for any $\epsilon>0$, there is an $\epsilon$-approximate sampling algorithm for $\mu_{\rho_\mathcal{S}}$ with expected runtime $O(\abs{G}\log(\abs{G}/\epsilon))$ and an $\epsilon$-approximate counting algorithm for $Z_\mathcal{S}$ with expected runtime $O(\abs{G}^2\epsilon^{-2}\log(\abs{G}/\epsilon)^2)$.
\end{theorem}

We prove Theorem~\ref{theorem:AntiferromagneticBipartiteHeisenbergModelAlgorithm} by showing that the conditions required to apply Theorem~\ref{theorem:PolymerSamplingAlgorithm} and Theorem~\ref{theorem:PolymerCountingAlgorithm} are satisfied. In particular, we show that (1) the antiferromagnetic Heisenberg model on bipartite graphs admits a suitable subgraph polymer model representation, (2) the Heisenberg polymer model admits a suitable probabilistic representation, and (3) the probabilistic Heisenberg polymer model is $\kappa$-computable.

We now introduce a polymer model representation of the Heisenberg model based on the geometric representation of Aizenman and Nachtergaele~\cite{aizenman1994geometric, goldschmidt2011quantum}. Let $\mathcal{S}=(\beta,h)$ be an antiferromagnetic Heisenberg model on a bipartite graph $G$. Let $V_A \sqcup V_B$ be a bipartition of $V(G)$, and let $\epsilon_v\coloneqq+1$ for $v \in V_A$ and $\epsilon_v\coloneqq-1$ for $v \in V_B$. Further, for each $e \in E(G)$, let $P_e$ denote the operator on the spins at the endpoints of $e$ defined by $P_e\coloneqq2\ketbra{s_e}{s_e}$, where $\ket{s_e}\coloneqq\frac{1}{\sqrt{2}}\left(\ket{+-}-\ket{-+}\right)$, and note that $T_e=\mathbb{I}-P_e$. For a sequence of edges $\rho=(\{u_\tau,v_\tau\})_{\tau\in\mathbb{Z}_n}$ of a connected subgraph $\gamma$ of $G$, we define $H_\rho$ to be the graph whose vertex set is $V(\gamma)\times\mathbb{Z}_n$ and whose edge set is $E_\rho^\cap \cup E_\rho^\cup \cup E_\rho^\|$, where $E_\rho^\cap\coloneqq\{\{(u_\tau,\tau),(v_\tau,\tau)\}\mid\tau\in\mathbb{Z}_n\}$, $E_\rho^\cup\coloneqq\{\{(u_\tau,\tau+1),(v_\tau,\tau+1)\}\mid\tau\in\mathbb{Z}_n\}$, and $E_\rho^\|\coloneqq\{\{(v,\tau),(v,\tau+1)\}\mid\tau\in\mathbb{Z}_n,v \in V(\gamma){\setminus}\{u_\tau,v_\tau\}\}$. Let $L_\rho$ denote the connected components of $H_\rho$, which we call \emph{loops}. For each loop $l \in L_\rho$, let $l_0\coloneqq\{v \in V(\gamma) \mid (v,0) \in V(l)\}$. A loop $l \in L_\rho$ is called \emph{internal} if $l_0=\varnothing$, and \emph{external} otherwise. Let $L_\rho^\text{int}$ and $L_\rho^\text{ext}$ denote the sets of internal and external loops, respectively. For each loop $l \in L_\rho$, we define its \emph{winding number} by $w(l)\coloneqq\sum_{v \in l_0}\epsilon_v$, and note that $w(l)=0$ for all $l \in L_\rho^\text{int}$.

We define the \emph{Heisenberg polymer model} of $\mathcal{S}$ to be the subgraph polymer model $\mathcal{P}=(\mathcal{C},w,\sim)$ on $G$, where $\mathcal{C}$ is the set of all connected subgraphs of $G$ with at least one edge, and $w:\mathcal{C}\to\mathbb{R}_{\geq0}$ is defined by
\begin{equation}
    w_\gamma \coloneqq \sum_{\rho \in P_\gamma}\frac{\beta^\abs{\rho}2^\abs{L_\rho^\text{int}}\prod_{l \in L_\rho^\text{ext}}2\cosh(\beta hw(l))}{\abs{\rho}!\left(2\cosh(\beta h)\right)^\abs{\gamma}}, \notag
\end{equation}
where $P_\gamma$ denotes the set of all sequences $\rho$ of edges of $\gamma$ with $\operatorname{supp}(\rho)=E(\gamma)$. For each $\gamma\in\mathcal{C}$ with $w_\gamma>0$, let $\mu_\gamma$ be the probability distribution on $\{-1,+1\}^{V(\gamma)}$ defined by
\begin{equation}
    \mu_\gamma(\sigma) \coloneqq \frac{1}{w_\gamma}\sum_{\rho \in P_\gamma}\frac{\beta^\abs{\rho}2^\abs{L_\rho^\text{int}}\prod_{l \in L_\rho^\text{ext}}e^{\beta h\sum_{v \in l_0}\sigma_v}\mathds{1}_\text{$\epsilon_{l_0}\sigma_{l_0}$ constant}}{\abs{\rho}!\left(2\cosh(\beta h)\right)^\abs{\gamma}}. \notag
\end{equation}
For each $v \in V(G)$, let $\mu_v$ be the probability distribution on $\{-1,+1\}$ defined by $\mu_v(\sigma)=\frac{e^{\beta h\sigma}}{2\cosh(\beta h)}$ for all $\sigma\in\{-1,+1\}$.

\begin{lemma}
    \label{lemma:AntiferromagneticBipartiteHeisenbergPolymerRepresentation}
    Let $\mathcal{S}=(\beta,h)$ be an antiferromagnetic Heisenberg model on a bipartite graph $G$, and let $\mathcal{P}=(\mathcal{C},w,\sim)$ be the Heisenberg polymer model of $\mathcal{S}$. Then the partition function of $\mathcal{S}$ satisfies $Z_\mathcal{S}=\cosh(\beta h)^\abs{G}e^{-\beta\norm{G}}Z_\mathcal{P}$ and the thermal distribution of $\mathcal{S}$ satisfies $\mu_{\rho_\mathcal{S}}=\hat{\mu}_\mathcal{P}$.
\end{lemma}

\begin{proof}
    We first show that $Z_\mathcal{S}=\cosh(\beta h)^\abs{G}e^{-\beta\norm{G}}Z_\mathcal{P}$.
    By the principle of inclusion-exclusion and commutativity,
    \begin{align}
        Z_\mathcal{S} &= \Tr\left[e^{-\beta H_\mathcal{S}}\right] \notag \\
        &= e^{-\beta\norm{G}}\sum_{S \subseteq E(G)}(-1)^\abs{S}\sum_{T \subseteq S}(-1)^\abs{T}\Tr\left[e^{\beta h\sum_{v \in V(G)}Z_v}e^{\beta\sum_{e \in T}P_e}\right]. \notag
    \end{align}
    For a subset $S \subseteq E(G)$, let $\Gamma_S$ denote the maximally connected components of $S$. By factorising over these components,
    \begin{align}
        Z_\mathcal{S} &= e^{-\beta\norm{G}}\Tr\left[e^{\beta h\sum_{v \in V(G)}Z_v}\right]\sum_{S \subseteq E(G)}\prod_{\gamma\in\Gamma_S}(-1)^\norm{\gamma}\sum_{T \subseteq E(\gamma)}(-1)^\abs{T}\frac{\Tr\left[e^{\beta h\sum_{v \in V(\gamma)}Z_v}e^{\beta\sum_{e \in T}P_e}\right]}{\Tr\left[e^{\beta h\sum_{v \in V(\gamma)}Z_v}\right]} \notag \\
        &= \cosh(\beta h)^\abs{G}e^{-\beta\norm{G}}\sum_{S \subseteq E(G)}\prod_{\gamma\in\Gamma_S}(-1)^\norm{\gamma}\sum_{T \subseteq E(\gamma)}(-1)^\abs{T}\frac{\Tr\left[e^{\beta h\sum_{v \in V(\gamma)}Z_v}e^{\beta\sum_{e \in T}P_e}\right]}{\cosh(\beta h)^\abs{\gamma}}. \notag
    \end{align}
    By the Taylor series and the principle of inclusion-exclusion,
    \begin{align}
        Z_\mathcal{S} &= \cosh(\beta h)^\abs{G}e^{-\beta\norm{G}}\sum_{S \subseteq E(G)}\prod_{\gamma\in\Gamma_S}\sum_{n=\norm{\gamma}}^\infty\frac{\beta^n}{n!}\sum_{\substack{\rho \in E(\gamma)^n \\ \operatorname{supp}(\rho)=E(\gamma)}}\frac{\Tr\left[e^{\beta h\sum_{v \in V(\gamma)}Z_v}\prod_{e\in\rho}P_e\right]}{\cosh(\beta h)^\abs{\gamma}} \notag \\
        &= \cosh(\beta h)^\abs{G}e^{-\beta\norm{G}}\sum_{S \subseteq E(G)}\prod_{\gamma\in\Gamma_S}\sum_{\rho \in P_\gamma}\frac{\beta^\abs{\rho}\Tr\left[e^{\beta h\sum_{v \in V(\gamma)}Z_v}\prod_{e\in\rho}P_e\right]}{\abs{\rho}!\cosh(\beta h)^\abs{\gamma}}. \notag
    \end{align}
    Since the trace factorises over the loops $L_\rho$, with each loop $l \in L_\rho^\text{int}$ contributing a factor of $2$, and each loop $l \in L_\rho^\text{ext}$ contributing a non-zero factor of $2\cosh(\beta hw(l))$ only when $\epsilon_v\sigma_v$ is constant on $l_0$,
    \begin{align}
         Z_\mathcal{S} &= \cosh(\beta h)^\abs{G}e^{-\beta\norm{G}}\sum_{S \subseteq E(G)}\prod_{\gamma\in\Gamma_S}\sum_{\rho \in P_\gamma}\frac{\beta^\abs{\rho}2^\abs{L_\rho^\text{int}}\prod_{l \in L_\rho^\text{ext}}2\cosh(\beta hw(l))}{\abs{\rho}!\left(2\cosh(\beta h)\right)^\abs{\gamma}} \notag \\
         &= \cosh(\beta h)^\abs{G}e^{-\beta\norm{G}}\sum_{S \subseteq E(G)}\prod_{\gamma\in\Gamma_S}w_\gamma \notag \\
         &= \cosh(\beta h)^\abs{G}e^{-\beta\norm{G}}\sum_{\Gamma\in\mathcal{G}}\prod_{\gamma\in\Gamma}w_\gamma \notag \\
         &= \cosh(\beta h)^\abs{G}e^{-\beta\norm{G}}Z_\mathcal{P}. \notag
    \end{align}
    We now show that $\mu_{\rho_\mathcal{S}}=\hat{\mu}_\mathcal{P}$. For any spin configuration $\sigma\in\{-1,+1\}^{V(G)}$, we have
    \begin{align}
        \mu_{\rho_\mathcal{S}}(\sigma) &= \frac{1}{Z_\mathcal{S}}\Tr\left[\ketbra{\sigma}{\sigma}e^{-\beta H_\mathcal{S}}\right] \notag \\
        &= \frac{1}{Z_\mathcal{P}}\sum_{\Gamma\in\mathcal{G}}\frac{e^{\beta h\sum_{v \in V(G)}\sigma_v}}{\left(2\cosh(\beta h)\right)^\abs{G}}\prod_{\gamma\in\Gamma}\sum_{\rho \in P_\gamma}\frac{2^\abs{\gamma}\beta^\abs{\rho}}{\abs{\rho}!}\Tr\left[\ketbra{\sigma_{V(\gamma)}}{\sigma_{V(\gamma)}}\prod_{e\in\rho}P_e\right] \notag \\
        &= \frac{1}{Z_\mathcal{P}}\sum_{\Gamma\in\mathcal{G}}\frac{e^{\beta h\sum_{v \in V(G)}\sigma_v}}{\left(2\cosh(\beta h)\right)^\abs{G}}\prod_{\gamma\in\Gamma}\sum_{\rho \in P_\gamma}\frac{\beta^\abs{\rho}2^\abs{L_\rho^\text{int}}}{\abs{\rho}!}\prod_{l \in L_\rho^\text{ext}}\mathds{1}_\text{$\epsilon_{l_0}\sigma_{l_0}$ constant} \notag \\
        &= \frac{1}{Z_\mathcal{P}}\sum_{\Gamma\in\mathcal{G}}\frac{e^{\beta h\sum_{v \in V(G)}\sigma_v}}{\left(2\cosh(\beta h)\right)^\abs{G}}\prod_{\gamma\in\Gamma}\frac{\left(2\cosh(\beta h)\right)^\abs{\gamma}}{e^{\beta h\sum_{v \in V(\gamma)}\sigma_v}}w_\gamma\mu_\gamma(\sigma_{V(\gamma)}) \notag \\
        &= \sum_{\Gamma\in\mathcal{G}}\mu_\mathcal{P}(\Gamma)\prod_{\gamma\in\Gamma}\mu_\gamma(\sigma_{V(\gamma)})\prod_{v \in V(G){\setminus}V(\Gamma)}\mu_v(\sigma_v) \notag \\
        &= \hat{\mu}_\mathcal{P}(\sigma). \notag
    \end{align}
    Finally, we show that $w_\gamma\geq0$ for all $\gamma\in\mathcal{C}$, and that $\mu_\gamma$ is a valid probability distribution on $\{-1,+1\}^{V(\gamma)}$ for all $\gamma\in\mathcal{C}$ with $w_\gamma>0$. Since all terms in $w_\gamma$ and $\mu_\gamma$ are non-negative, we have $w_\gamma\geq0$ for all $\gamma\in\mathcal{C}$, and $\mu_\gamma(\sigma)\geq0$ for all $\gamma\in\mathcal{C}$ with $w_\gamma>0$ and $\sigma\in\{-1,+1\}^{V(\gamma)}$. By using $\sum_{\sigma\in\{-1,+1\}^{l_0}}e^{\beta h\sum_{v \in l_0}\sigma_v}\mathds{1}_\text{$\epsilon_{l_0}\sigma_{l_0}$ constant}=2\cosh(\beta hw(l))$ for each $l \in L_\rho^\text{ext}$, and $w(l)=0$ for all $l \in L_\rho^\text{int}$, we have $\sum_{\sigma\in\{-1,+1\}^{V(\gamma)}}\mu_\gamma(\sigma)=1$, and hence $\mu_\gamma$ defines a valid probability distribution on $\{-1,+1\}^{V(\gamma)}$. This completes the proof.
\end{proof}

Let $\mathcal{S}=(\beta,h)$ be an antiferromagnetic Heisenberg model on a bipartite graph $G$, and let $\mathcal{P}=(\mathcal{C},w,\sim)$ be the Heisenberg polymer model of $\mathcal{S}$. For each $\gamma\in\mathcal{C}$, let $\nu_\gamma$ denote the probability distribution over sequences $\rho$ of edges of $\gamma$ with $\operatorname{supp}(\rho)=E(\gamma)$ defined by sampling $n\sim\operatorname{Poi}(2\beta\norm{\gamma})$ and then $\rho$ uniformly at random from sequences of length $\norm{\gamma}+n$ with support $E(\gamma)$. We define the \emph{probabilistic Heisenberg polymer model} of $\mathcal{P}$ to be the probabilistic polymer model $\mathcal{P}_\circ=(\mathcal{C},\tilde{w},\sim)$ where, for each $\gamma\in\mathcal{C}$, the weight $\tilde{w}_\gamma$ is defined for $\rho\sim\nu_\gamma$ by
\begin{equation}
    \tilde{w}_\gamma(\rho) \coloneqq \frac{(2\beta e^{2\beta})^\norm{\gamma}\genfrac{\{}{\}}{0pt}{}{\abs{\rho}}{\norm{\gamma}}\norm{\gamma}^\norm{\gamma}}{\binom{\abs{\rho}}{\norm{\gamma}}(2\norm{\gamma})^\abs{\rho}}\frac{2^\abs{L_\rho^\text{int}}\prod_{l \in L_\rho^\text{ext}}2\cosh(\beta hw(l))}{\left(2\cosh(\beta h)\right)^\abs{\gamma}}. \notag
\end{equation}
For each $\gamma\in\mathcal{C}$ and $\rho\in\operatorname{supp}(\nu_\gamma)$ with $\tilde{w}_\gamma(\rho)>0$, let $\tilde{\mu}_{\gamma,\rho}$ be the probability distribution on $\{-1,+1\}^{V(\gamma)}$ defined by
\begin{equation}
    \tilde{\mu}_{\gamma,\rho}(\sigma) \coloneqq \prod_{l \in L_\rho^\text{ext}}\frac{e^{\beta h\sum_{v \in l_0}\sigma_v}\mathds{1}_\text{$\epsilon_{l_0}\sigma_{l_0}$ constant}}{2\cosh(\beta hw(l))}. \notag
\end{equation}

\begin{lemma}
    \label{lemma:AntiferromagneticBipartiteHeisenbergWeightEstimator}
    Fix $\Delta\in\mathbb{Z}_{\geq3}$ and $r\geq0$. Let $\mathcal{S}=(\beta,h)$ be an antiferromagnetic Heisenberg model on a bipartite graph $G$ of maximum degree at most $\Delta$, let $\mathcal{P}=(\mathcal{C},w,\sim)$ be the Heisenberg polymer model of $\mathcal{S}$, and let $\mathcal{P}_\circ=(\mathcal{C},\tilde{w},\sim)$ be the probabilistic Heisenberg polymer model of $\mathcal{P}$. Let $\beta$ be such that
    \begin{equation}
        \beta \leq \frac{1}{2e^{r+3}\Delta}. \notag
    \end{equation}
    Then, for all $\gamma\in\mathcal{C}$ and $\rho\in\operatorname{supp}(\nu_\gamma)$,
    \begin{equation}
        \mathbb{E}_{\nu_\gamma}\left[\tilde{w}_\gamma\right]=w_\gamma \quad\text{and}\quad \tilde{w}_\gamma(\rho) \leq \left(\frac{1}{2e^{r+2}(\Delta-1)}\right)^\norm{\gamma}, \notag
    \end{equation}
    and, for all $\gamma\in\mathcal{C}$ with $w_\gamma>0$ and $\sigma\in\{-1,+1\}^{V(\gamma)}$,
    \begin{equation}
        \mathbb{E}_{\nu_\gamma}\left[\tilde{w}_\gamma\tilde{\mu}_{\gamma,\rho}(\sigma)\right] = w_\gamma\mu_\gamma(\sigma). \notag
    \end{equation}
\end{lemma}

\begin{proof}
    Fix a polymer $\gamma\in\mathcal{C}$. We first show that $\mathbb{E}_{\nu_\gamma}\left[\tilde{w}_\gamma\right]=w_\gamma$. Since there are $\genfrac{\{}{\}}{0pt}{}{\norm{\gamma}+n}{\norm{\gamma}}\norm{\gamma}!$ sequences $\rho$ of length $\norm{\gamma}+n$ whose support is $E(\gamma)$,
    \begin{align}
        \mathbb{E}_{\nu_\gamma}[\tilde{w}_\gamma] &= \sum_{n=0}^\infty\frac{(2\beta\norm{\gamma})^ne^{-2\beta\norm{\gamma}}}{n!}\frac{1}{\genfrac{\{}{\}}{0pt}{}{\norm{\gamma}+n}{\norm{\gamma}}\norm{\gamma}!}\sum_{\substack{\rho \in E(\gamma)^{\norm{\gamma}+n} \\ \operatorname{supp}(\rho)=E(\gamma)}}\tilde{w}_\gamma(\rho) \notag \\
        &= \sum_{n=0}^\infty\frac{\beta^{\norm{\gamma}+n}}{(\norm{\gamma}+n)!}\sum_{\substack{\rho \in E(\gamma)^{\norm{\gamma}+n} \\ \operatorname{supp}(\rho)=E(\gamma)}}\frac{2^\abs{L_\rho^\text{int}}\prod_{l \in L_\rho^\text{ext}}2\cosh(\beta hw(l))}{\left(2\cosh(\beta h)\right)^\abs{\gamma}} \notag \\
        &= \sum_{\rho \in P_\gamma}\frac{\beta^\abs{\rho}}{\abs{\rho}!}\frac{2^\abs{L_\rho^\text{int}}\prod_{l \in L_\rho^\text{ext}}2\cosh(\beta hw(l))}{\left(2\cosh(\beta h)\right)^\abs{\gamma}} \notag \\
        &= w_\gamma. \notag
    \end{align}
    We now show that $\tilde{w}_\gamma(\rho)\leq\left(\frac{1}{2e^{r+2}(\Delta-1)}\right)^\norm{\gamma}$ for all $\rho\in\operatorname{supp}(\nu_\gamma)$. Since $\abs{w(l)}\leq\abs{l_0}$ for each $l \in L_\rho^\text{ext}$, we have $2\cosh(\beta hw(l))\leq\left(2\cosh(\beta h)\right)^\abs{l_0}$ by using $2\cosh(nx)\leq\left(2\cosh(x)\right)^n$ for all $x\in\mathbb{R}$ and $n\in\mathbb{Z}^+$. Further, $\sum_{l \in L_\rho^\text{ext}}\abs{l_0}=\abs{\gamma}$ and $\abs{L_\rho^\text{int}}\leq\abs{\rho}$. Therefore,
    \begin{align}
        \tilde{w}_\gamma(\rho) &= \frac{(2\beta e^{2\beta})^\norm{\gamma}\genfrac{\{}{\}}{0pt}{}{\abs{\rho}}{\norm{\gamma}}\norm{\gamma}^\norm{\gamma}}{\binom{\abs{\rho}}{\norm{\gamma}}(2\norm{\gamma})^\abs{\rho}}\frac{2^\abs{L_\rho^\text{int}}\prod_{l \in L_\rho^\text{ext}}2\cosh(\beta hw(l))}{\left(2\cosh(\beta h)\right)^\abs{\gamma}} \notag \\
        &\leq \frac{(2\beta e^{2\beta})^\norm{\gamma}\genfrac{\{}{\}}{0pt}{}{\abs{\rho}}{\norm{\gamma}}\norm{\gamma}^\norm{\gamma}}{\binom{\abs{\rho}}{\norm{\gamma}}\norm{\gamma}^\abs{\rho}} \notag \\
        &\leq (2\beta e^{2\beta})^\norm{\gamma}, \notag
    \end{align}
    where we have used $\genfrac{\{}{\}}{0pt}{}{\abs{\rho}}{\norm{\gamma}}\norm{\gamma}^\norm{\gamma}\leq\binom{\abs{\rho}}{\norm{\gamma}}\norm{\gamma}^\abs{\rho}$. By taking $\beta\leq\frac{1}{2e^{r+3}\Delta}$,
    \begin{equation}
        \tilde{w}_\gamma(\rho) \leq \left(\frac{1}{2e^{r+2}(\Delta-1)}\right)^\norm{\gamma}. \notag
    \end{equation}
    Finally, we show that $\mathbb{E}_{\nu_\gamma}\left[\tilde{w}_\gamma\tilde{\mu}_{\gamma,\rho}(\sigma)\right]=w_\gamma\mu_\gamma(\sigma)$ for all $\gamma\in\mathcal{C}$ with $w_\gamma>0$ and $\sigma\in\{-1,+1\}^{V(\gamma)}$. By a similar argument as for $\mathbb{E}_{\nu_\gamma}[\tilde{w}_\gamma]=w_\gamma$,
    \begin{align}
        \mathbb{E}_{\nu_\gamma}\left[\tilde{w}_\gamma\tilde{\mu}_{\gamma,\rho}(\sigma)\right] &= \sum_{\rho \in P_\gamma}\frac{\beta^\abs{\rho}2^\abs{L_\rho^\text{int}}\prod_{l \in L_\rho^\text{ext}}e^{\beta h\sum_{v \in l_0}\sigma_v}\mathds{1}_\text{$\epsilon_{l_0}\sigma_{l_0}$ constant}}{\abs{\rho}!\left(2\cosh(\beta h)\right)^\abs{\gamma}} \notag \\
        &= w_\gamma\mu_\gamma(\sigma). \notag
    \end{align}
    This completes the proof.
\end{proof}

\begin{lemma}
    \label{lemma:ProbabilisticAntiferromagneticBipartiteHeisenbergKappaComputable}
    Let $\mathcal{S}=(\beta,h)$ be an antiferromagnetic Heisenberg model on a bipartite graph $G$, let $\mathcal{P}=(\mathcal{C},w,\sim)$ be the Heisenberg polymer model of $\mathcal{S}$, and let $\mathcal{P}_\circ=(\mathcal{C},\tilde{w},\sim)$ be the probabilistic Heisenberg polymer model of $\mathcal{P}$. Then $\mathcal{P}_\circ$ is $\kappa$-computable with $\kappa=0$.
\end{lemma}

\begin{proof}
    A sample $\rho\sim\nu_\gamma$ can be obtained in expected time $\norm{\gamma}^{O(1)}$ by sampling $n\sim\operatorname{Poi}(2\beta\norm{\gamma})$ and then sampling a uniform surjection of length $\norm{\gamma}+n$ onto $E(\gamma)$. For such a sample $\rho$, the loop decomposition $L_\rho$ can be computed in expected time $\norm{\gamma}^{O(1)}$. Hence, $\tilde{w}_\gamma(\rho)$ can be computed and $\tilde{\mu}_{\gamma,\rho}$ can be sampled in expected time $\norm{\gamma}^{O(1)}$, and so $\mathcal{P}_\circ$ is $\kappa$-computable with $\kappa=0$.
\end{proof}

We now prove Theorem~\ref{theorem:AntiferromagneticBipartiteHeisenbergModelAlgorithm}.

\begin{proof}[Proof of Theorem~\ref*{theorem:AntiferromagneticBipartiteHeisenbergModelAlgorithm}]
    Let $\mathcal{P}=(\mathcal{C},w,\sim)$ be the Heisenberg polymer model of $\mathcal{S}$ and let $\mathcal{P}_\circ=(\mathcal{C},\tilde{w},\sim)$ be the probabilistic Heisenberg polymer model of $\mathcal{P}$. By Lemma~\ref{lemma:AntiferromagneticBipartiteHeisenbergPolymerRepresentation}, $\mu_{\rho_\mathcal{S}}=\hat{\mu}_\mathcal{P}$ and $Z_\mathcal{S}=\cosh(\beta h)^\abs{G}e^{-\beta\norm{G}}Z_\mathcal{P}$, so it suffices to sample from $\hat{\mu}_\mathcal{P}$ and approximate $Z_\mathcal{P}$. A sample from $\hat{\mu}_\mathcal{P}$ is obtained by sampling $\Gamma\sim\mu_\mathcal{P}$ and then sampling spin configurations $\sigma_{V(\gamma)}\sim\tilde{\mu}_{\gamma,\rho_\gamma}$ for each $\gamma\in\Gamma$ and $\sigma_v\sim\mu_v$ for each $v \in V(G){\setminus}V(\Gamma)$, where $\rho_\gamma$ denotes the sequence sampled by the accepting iteration of the graphlet sampler that produced $\gamma$. By Lemma~\ref{lemma:AntiferromagneticBipartiteHeisenbergWeightEstimator} with $r=0$, for all $\gamma\in\mathcal{C}$ and $\rho\in\operatorname{supp}(\nu_\gamma)$,
    \begin{equation}
        \mathbb{E}_{\nu_\gamma}\left[\tilde{w}_\gamma\right]=w_\gamma \quad\text{and}\quad \tilde{w}_\gamma(\rho) \leq \left(\frac{1}{2e^2(\Delta-1)}\right)^\norm{\gamma}, \notag
    \end{equation}
    and, for all $\gamma\in\mathcal{C}$ with $w_\gamma>0$ and $\sigma\in\{-1,+1\}^{V(\gamma)}$,
    \begin{equation}
        \mathbb{E}_{\nu_\gamma}\left[\tilde{w}_\gamma\tilde{\mu}_{\gamma,\rho}(\sigma)\right] = w_\gamma\mu_\gamma(\sigma). \notag
    \end{equation}
    In particular, $\mathcal{P}_\circ$ satisfies the conditions of Lemma~\ref{lemma:ProbabilisticGraphletSampler}. Since $\mathcal{P}$ is a subgraph polymer model on $G$ of maximum degree at most $\Delta$, it is equivalent to a graphlet polymer model on $L(G)$, which has maximum degree at most $2(\Delta-1)$. By Lemma~\ref{lemma:ProbabilisticAntiferromagneticBipartiteHeisenbergKappaComputable}, $\mathcal{P}_\circ$ is $\kappa$-computable with $\kappa=0$. By applying Theorem~\ref{theorem:PolymerSamplingAlgorithm} and Theorem~\ref{theorem:PolymerCountingAlgorithm} with $r=\kappa$, there is an $\epsilon$-approximate sampling algorithm for $\hat{\mu}_\mathcal{P}$ with expected runtime $O(\abs{G}\log(\abs{G}/\epsilon))$ and an $\epsilon$-approximate counting algorithm for $Z_\mathcal{P}$ with expected runtime $O(\abs{G}^2\epsilon^{-2}\log(\abs{G}/\epsilon)^2)$. This completes the proof.
\end{proof}

\section{Conclusion \& Outlook}
\label{section:ConclusionOutlook}

We have established a general framework for developing fast sampling and counting algorithms for stoquastic spin systems at high temperature, based on the polymer dynamics of Chen et al.~\cite{chen2021fast} and the graphlet sampling algorithm of Blanca et al.~\cite{blanca2024fast}. We obtained fast algorithms for approximating the partition function and sampling from the thermal distribution of (1) general stoquastic spin systems, (2) ferromagnetic Heisenberg models, and (3) antiferromagnetic Heisenberg models on bipartite graphs. For the Heisenberg models, we obtained an improved bound on the inverse temperature by using the cycle representation of T\'oth~\cite{toth1993improved, goldschmidt2011quantum} and the loop representation of Aizenman and Nachtergaele~\cite{aizenman1994geometric, goldschmidt2011quantum} to construct an unbiased estimator for the polymer weights that can be computed in polynomial time.

Several natural questions remain open. First, it would be interesting to obtain optimal bounds on the inverse temperature for which our algorithms apply. Second, it would be interesting to extend our results to the low-temperature setting and to graphs of unbounded degree, both of which have been achieved for classical models~\cite{chen2021fast, galanis2022fast, blanca2024fast}. Finally, it would be of significant interest to obtain fast classical algorithms for general quantum spin systems beyond the stoquastic setting.

\section*{Acknowledgements}

We thank Tyler Helmuth, Will Perkins, and Gabriel Waite for helpful discussions. RLM was supported by the ARC Training Centre for Future Leaders in Quantum Computing (FLiQC), project number IC240100009.

\bibliography{bibliography}

\end{document}